\documentclass[11pt,a4paper]{article}
\usepackage[utf8]{inputenc}
\usepackage[titletoc,title]{appendix}
\usepackage{amsmath,amsfonts,amssymb,bm,hyperref,amsthm}
\usepackage{enumerate}
\numberwithin{equation}{section}
\newtheorem{theorem}{Theorem}[section]

\newtheorem{lemma}[theorem]{Lemma}
\newtheorem{fact}{Fact}[subsection]

\theoremstyle{definition}
\newtheorem{example}[theorem]{Example}
\newtheorem{remark}[theorem]{Remark}
\newtheorem{definition}[theorem]{Definition}

\numberwithin{equation}{section}
\numberwithin{equation}{subsection}
\renewcommand*{\theequation}{%
  \ifnum\value{subsection}=0 %
    \thesection
  \else
    \thesubsection
  \fi
  .\arabic{equation}%
}

\usepackage{color}
\usepackage[all,color]{xy}

\usepackage[top=1in, bottom=1in, left=1in, right=1in]{geometry}
\usepackage[demo]{graphicx}
\usepackage{caption}
\usepackage{subcaption}

\usepackage{pgfplots}
\usepackage{tikz}

\begin{document}

\title{Spacetime diffeomorphisms as matter fields}

\author{
Matteo Capoferri\thanks{MC:
Department of Mathematics,
University College London,
Gower Street,
London WC1E~6BT,
UK;
matteo.capoferri@gmail.com.
Current address: School of Mathematics, Cardiff University, Senghennydd Road, Cardiff CF24 4AG, UK.
}
\and
Dmitri Vassiliev\thanks{DV:
Department of Mathematics,
University College London,
Gower Street,
London WC1E~6BT,
UK;
d.vassiliev@ucl.ac.uk,
\url{http://www.ucl.ac.uk/\~ucahdva/};
DV was supported by EPSRC grant EP/M000079/1.
}}


\renewcommand\footnotemark{}


\maketitle
\begin{abstract}
We work on a 4-manifold equipped with Lorentzian metric $g$ and consider a volume-preserving diffeomorphism
which is the unknown quantity of our mathematical model.
The diffeomorphism defines a second Lorentzian metric $h$, the pullback of $g$.
Motivated by elasticity theory,
we introduce a Lagrangian
expressed algebraically (without differentiations) via our pair of metrics.
Analysis of the resulting nonlinear field equations
produces three main results.
Firstly, we show that for Ricci-flat manifolds our linearised
field equations are Maxwell's equations in the Lorenz gauge with exact current.
Secondly, for Minkowski space we construct explicit massless solutions of our nonlinear field equations;
these come in two distinct types,
right-handed and left-handed.
Thirdly, for Minkowski space we construct explicit massive solutions of our nonlinear field equations;
these contain a positive parameter
which has the geometric meaning of quantum mechanical mass and a real
parameter which may be interpreted as electric charge.
In constructing explicit solutions of nonlinear field equations
we resort to group-theoretic ideas: we identify special 4-dimensional subgroups
of the Poincar\'e group and seek diffeomorphisms compatible with their action, in a suitable sense.

\


\

{\bf MSC classes: }
primary
53C50
; secondary 
22E43
, 35Q41
, 35Q61
, 74B20.

\end{abstract}


\tableofcontents


\section{Introduction}

In this paper we propose a new mathematical model for a class of field theories in the Lorentzian setting. Inspired by the classical theory of elasticity (see, e.g., \cite{ciarlet1994,ciarlet2005}), we construct a Lagrangian out of a pair of metrics related by a spacetime diffeomorphism, which, in turn, represents the unknown of our model. The variation of our Lagrangian under the volume preservation condition produces a system of nonlinear partial differential equations, the field equations, whose analysis constitutes the main goal of the paper.

Our work possesses several elements of novelty. Firstly, in spite of relying on ideas from Riemannian elasticity, our theory is fully Lorentzian in that it deals with diffeomorphisms of the whole spacetime into itself, giving detailed account of the issues arising due to the indefinite signature. Secondly, our model incorporates a volume preservation condition into a theory of elasticity, leading to interesting mathematical consequences 
--- see, e.g., \eqref{incompressibility 2} and \eqref{linearisation of the volume preservation condition}. 
Thirdly, we suggest new techniques for solving nonlinear PDEs, ones of possibly broader relevance. Lastly, our construction gives rise to solutions that appear to be physically meaningful, with potential applications in the realm of theoretical and particle physics.

For the case of Minkowski spacetime, we provide two classes of explicit solutions,
massless and massive,
which, at least at a formal level, offer a natural physical interpretation
in terms of elementary particles, namely, neutrino/antineutrino
and electron/positron.
Our massive solution contains two free parameters.
Even though these parameters can be interpreted
as quantum mechanical mass and electric charge,
our model does not allow for their values to be determined.
We attribute this to the large number of symmetries implicitly present in our theory.
One would hope that appropriate symmetry breaking could
overcome this shortcoming of our mathematical model.

Our model is, effectively, a nonlinear version of Maxwell's theory.
The only dimensional parameter is the speed of light: it is encoded in the Minkowski
metric when we consider the case of flat spacetime.
All other parameters are dimensionless
and are contained in our Largrangian.

We develop our theory in dimension four and for pseudo-Riemannian manifolds of Lorentzian 
signature. In principle, neither assumption is necessary for its formulation.
However, the physical conclusions we derive are specific to dimension $3+1$.
In particular, dimension $3+1$ appears to be the lowest in which one observes
propagating massless solutions.

The paper is structured as follows.
In Section~\ref{Mathematical model} we present the mathematical formulation of our model.
In Section~\ref{Nonlinear field equations}
we derive the corresponding nonlinear field equations,
accounting for the volume preservation condition.
Section~\ref{Displacements and rotations}
is devoted to discussing the role of displacements and rotations;
in particular, we perform a detailed analysis of the deformation gradient
in terms of its Lorentzian polar decomposition.
Section~\ref{Linearised field equations} contains our first main result:
the linearised field equations and their connection with Maxwell's equations.
For Ricci-flat Lorentzian manifolds our model gives, in the linear approximation,
Maxwell's equations in the Lorenz gauge with exact current.
In Sections~\ref{Homogeneous diffeomorphisms}
and \ref{Special subgroups of the Poincare group}
we introduce the concept of homogeneous diffeomorphism
and special subgroups of the Poincar\'e group respectively.
These represent the group-theoretic tools
which lie at the foundation of our construction of  solutions to nonlinear PDEs.
Explicit solutions for Minkowski spacetime are presented in
Sections~\ref{Explicit massless solutions of nonlinear field equations}
and \ref{Explicit massive solutions of nonlinear field equations}.
Massless solutions described in
Section~\ref{Explicit massless solutions of nonlinear field equations}
come into two types:
right-handed and left-handed.
Massive solutions described in
Section~\ref{Explicit massive solutions of nonlinear field equations}
contain two free parameters:
a positive parameter which has the geometric meaning of quantum mechanical mass
and a real parameter which may be interpreted as electric charge.
Finally, in Section~\ref{Massless Dirac equation} and Section~\ref{Massive Dirac equation}
we present a formal argument, showing that our massless and massive solutions
can be associated with spinors satisfying
the massless and massive Dirac equations respectively.
This constitutes the first step towards possible future applications
of our model in theoretical and particle physics,
e.g.~in quantum electrodynamics,
see also \cite{sesqui} where a volume preservation condition appears in a somewhat different setting.
The paper is complemented by four appendices
dealing with notation and auxiliary technical results.

\section{Mathematical model}
\label{Mathematical model}

Let $M$ be a connected 4-manifold.
Local coordinates on $M$ will be denoted by $x=(x^1,x^2,x^3,x^4)$ or $y=(y^1,y^2,y^3,y^4)$.

We assume that our manifold $M$ is equipped with Lorentzian metric $g$
with signature
\linebreak
$+++-\,$.
Throughout this paper the metric $g$ is assumed to be prescribed.

The unknown quantity in  our mathematical model
is a diffeomorphism $\varphi:M\to M$.
We will denote the group of diffeomorphisms by $\mathrm{Diff}(M)$.

Let us introduce a new (perturbed) Lorentzian metric $h$ defined as the pullback of $g$ via $\varphi$,
$h:=\varphi^*g$.
In local coordinates this new metric is written as follows.
Take an arbitrary point $P\in M$ and choose local coordinates $x$ and $y$
in the neighbourhoods of $P$ and $\varphi(P)$ respectively.
Our diffeomorphism $\varphi$ can then be written locally as
\begin{equation}
\label{local representation of diffeomorphism}
y=\varphi(x).
\end{equation}
The new metric tensor reads
\begin{equation}
\label{definition of perturbed metric}
h_{\alpha\beta}(x):=
g_{\mu\nu}(\varphi(x))
\frac{\partial\varphi^\mu}{\partial x^\alpha}
\frac{\partial\varphi^\nu}{\partial x^\beta}\,.
\end{equation}
The $g_{\mu\nu}$ in the RHS of \eqref{definition of perturbed metric}
is the representation of the metric tensor $g$ in local coordinates $y$.

The following non-rigorous physical argument along
the lines of
\cite{MR0106584}
explains the geometric meaning of the tensor
\eqref{definition of perturbed metric}.
Consider two points, $x$~and $x+\Delta x$.
The interval
(Lorentzian analogue of `distance squared')
between these two points is $g_{\alpha\beta}(x)\,\Delta x^\alpha\Delta x^\beta$.
Our diffeomorphism maps $x$~and $x+\Delta x$ to
$\varphi(x)$ and $\varphi(x+\Delta x)\approx\varphi(x)+\frac{\partial\varphi}{\partial x^\alpha}\Delta x^\alpha$
respectively.
The interval between
$\varphi(x)$ and $\varphi(x)+\frac{\partial\varphi}{\partial x^\alpha}\Delta x^\alpha$
is
$
g_{\mu\nu}(\varphi(x))
\left(\frac{\partial\varphi^\mu}{\partial x^\alpha}\Delta x^\alpha\right)
\left(\frac{\partial\varphi^\nu}{\partial x^\beta}\Delta x^\beta\right)
$,
which, in view of \eqref{definition of perturbed metric},
can be rewritten concisely as $h_{\alpha\beta}(x)\,\Delta x^\alpha\Delta x^\beta$.
Therefore, the metric $h$ describes the interval between points of the
deformed continuum.

Having at our disposal two Lorentzian metrics, $g$ and $h$,
we can now write down an action.
To this end, let us first introduce some definitions.

\begin{definition}
The tensor
\begin{equation}
\label{definition of strain formula}
S^\alpha{}_\beta(x):=
[g^{\alpha\gamma}(x)]\,[h_{\gamma\beta}(x)]-\delta^\alpha{}_\beta
\end{equation}
is called \emph{strain}.
\end{definition}

The concept of strain tensor originates from the papers of
Cauchy \cite{cauchy1, cauchy2}.

The strain tensor describes, pointwise, a linear map in the fibres of the tangent bundle,
\begin{equation}
\label{strain as a linear map}
v^\alpha\mapsto S^\alpha{}_\beta\,v^\beta.
\end{equation}
The algebraic motivation for the introduction of the map
\eqref{strain as a linear map} is explained in
Appendix~\ref{Linear algebra involving a pair of quadratic forms}.

Let us now construct scalars out of a strain tensor.
This can be done in many different ways but only four, at most, will be
independent. An arbitrary scalar can be expressed, possibly in a nonlinear fashion,
via the four chosen independent scalars.
The obvious way of choosing four independent scalars is
$\operatorname{tr}(S^k)$, $k=1,2,3,4$,
but such a choice is inconvenient as it would make subsequent calculations cumbersome.
The most convenient choice of four scalar invariants is
\begin{subequations}\label{scalar invariants}
\begin{align}
\label{scalar invariant 1}
e_1(\varphi)&:=
\operatorname{tr}S,
\\
\label{scalar invariant 2}
e_2(\varphi)&:=
\frac12
\left[
(\operatorname{tr}S)^2-\operatorname{tr}(S^2)
\right],
\\
\label{scalar invariant 3}
e_3(\varphi)&:=
\operatorname{tr}\operatorname{adj}S,
\\
\label{scalar invariant 4}
e_4(\varphi)&:=
\det S.
\end{align}
\end{subequations}
Here $\,\operatorname{tr}\,$ is the matrix trace and $\,\operatorname{adj}\,$
is the operator of matrix adjugation from linear algebra.

The reasoning behind the particular choice
\eqref{scalar invariant 1}--\eqref{scalar invariant 4}
becomes clear if we rewrite these invariants in terms of the eigenvalues of strain.
The strain tensor \eqref{definition of strain formula},
viewed as a linear operator
\eqref{strain as a linear map}
acting in $\mathbb{C}^4$
has eigenvalues
$\lambda_k$, $k=1,2,3,4$,
enumerated with account of their algebraic multiplicity.
Note that some eigenvalues may be complex, in which case
they come in complex conjugate pairs. It is easy to see that formulae
\eqref{scalar invariant 1}--\eqref{scalar invariant 4} can be rewritten as
\begin{subequations}
\begin{align}
\label{scalar invariant lambda 1}
e_1(\varphi)&=
\lambda_1+\lambda_2+\lambda_3+\lambda_4\,,
\\
\label{scalar invariant lambda 2}
e_2(\varphi)&=
\lambda_1\lambda_2+\lambda_1\lambda_3+\lambda_1\lambda_4
+\lambda_2\lambda_3+\lambda_2\lambda_4+\lambda_3\lambda_4\,,
\\
\label{scalar invariant lambda 3}
e_3(\varphi)&=
\lambda_1\lambda_2\lambda_3
+\lambda_1\lambda_2\lambda_4
+\lambda_1\lambda_3\lambda_4
+\lambda_2\lambda_3\lambda_4\,,
\\
\label{scalar invariant lambda 4}
e_4(\varphi)&=
\lambda_1\lambda_2\lambda_3\lambda_4\,.
\end{align}
\end{subequations}
The advantage of choosing scalar invariants in this particular way
is that the polynomials appearing in the right-hand sides of formulae
\eqref{scalar invariant lambda 1}--\eqref{scalar invariant lambda 4}
are elementary symmetric polynomials.

Note that our scalars $\,e_k\,$
are spectral invariants:
we are looking at quantities that are determined by the spectrum of the linear map
\eqref{strain as a linear map}.
Our definition of scalar invariants is similar to that in
\cite[(3.56)]{spencer}, the only difference being that we have four scalar invariants
instead of three --- a consequence of us adopting a 4-dimensional relativistic approach.

Our action then is
\begin{equation}
\label{action 1}
\mathcal{J}(\varphi):=
\int_M\mathcal{L}\bigl(e_1(\varphi),e_2(\varphi),e_3(\varphi),e_4(\varphi)\bigr)
\,\sqrt{-\det g_{\mu\nu}(x)}
\ dx\,,
\end{equation}
where
$\mathcal{L}$ is some prescribed smooth real-valued function of four real variables
such that
\linebreak
$\mathcal{L}(0,0,0,0)=0$
and $dx:=dx^1dx^2dx^3dx^4$.
Variation of \eqref{action 1} with respect to the unknown diffeormorphism
$\varphi\in\mathrm{Diff}(M)$ generates field equations
which can be thought of as a Lorentzian version
of nonlinear elasticity.

The physical assumptions underlying our choice of action
\eqref{action 1} are isotropy and homo\-geneity of our 4-dimensional continuum.
Isotropy is expressed mathematically in that the integrand
$\mathcal{L}$ in \eqref{action 1} is a symmetric function of the eigenvalues
of the map \eqref{strain as a linear map}.
Homogeneity is expressed mathematically in that the integrand
$\mathcal{L}$ in \eqref{action 1} does not depend explicitly on $x$.

Two important examples of Lagrangians are given below.

\begin{example}[Linear Lagrangian]
\label{example of Linear Lagrangian}
The unique, up to rescaling, Lagrangian linear in strain is
\begin{equation}
\label{harmonic map action}
\mathcal{L}(e_1,e_2,e_3,e_4)=e_1.
\end{equation}
This is the action of a harmonic map,
see \cite{eells,wood},
the only difference being that
in our paper the metric is assumed to have Lorentzian signature.
\end{example}

\begin{example}[Quadratic Lagrangian]
\label{example of Quadratic Lagrangian}
The general form of a Lagrangian quadratic
(homogeneous of degree two) in strain is
\begin{equation}
\label{L mathcal quadratic}
\mathcal{L}(e_1,e_2,e_3,e_4)=\alpha(e_1)^2+\beta\,e_2\,,
\end{equation}
where $\alpha,\beta\in\mathbb{R}$ are parameters.
In the 3-dimensional Riemannian setting the above Lagrangian
is used in the theory of elasticity:
it describes a static isotropic homogeneous elastic continuum that
is physically linear but geometrically nonlinear.
The standard assumption in elasticity theory is
\begin{equation}
\label{beta nonzero}
\beta\ne0.
\end{equation}
Under the assumption \eqref{beta nonzero}
the Lagrangian \eqref{L mathcal quadratic} contains, up to rescaling,
only one dimensionless parameter, $\alpha/\beta$.
In elasticity theory the parameters
$\lambda=2\alpha+\beta$
and
$\mu=-\beta/2$
are called \emph{Lam\'e parameters}
and the parameter
$
\,\nu=\frac{2\alpha+\beta}{4\alpha+\beta}\,
$
is called \emph{Poisson's ratio}.
\end{example}

\begin{remark}
Our mathematical model does not involve the concepts of connection and curvature.
Moreover, it is easy to see that if the unperturbed metric $g$ is flat
then the perturbed metric $h$ is flat as well.
Our model is different from those commonly used in theories
of bimetric gravity
\cite{rosen,derham,hassan,schmidt},
even though the mathematical formalism is quite similar.
\end{remark}

Field equations
for the action \eqref{action 1}
are not the equations that we will be studying.
We choose to impose, in addition, the volume preservation constraint
\begin{equation}
\label{incompressibility 1}
\det g_{\alpha\beta}(x)=\det h_{\mu\nu}(x).
\end{equation}
In other words, we choose to restrict our
analysis to the subgroup of volume-preserving diffeomorphisms
$\mathrm{Diff}_\rho(M)\subset\mathrm{Diff}(M)$.
Here
\begin{equation}
\label{defiition of Lorentzian density}
\rho(x):=\sqrt{-\det g_{\alpha\beta}(x)}
\end{equation}
is the Lorentzian density of the unperturbed metric.

The condition for a diffeomorphism to be 
volume preserving reads, locally,
\begin{equation}
\label{definition of volume preservation}
\rho(x)=
\rho(\varphi(x))
\left|
\det\left(\frac{\partial\varphi^\alpha}{\partial x^\beta}\right)
\right|
.
\end{equation}
The $\rho$ in the LHS of \eqref{definition of volume preservation}
is the representation of the density $\rho$ in local coordinates $x$,
whereas
the $\rho$ in the RHS of \eqref{definition of volume preservation}
is the representation of the density $\rho$ in local coordinates~$y$,
see \eqref{local representation of diffeomorphism}.


In spectral-theoretic fashion,
the volume preservation constraint
\eqref{incompressibility 1} can be equivalently rewritten as
\begin{equation}
\label{incompressibility 2}
e_1(\varphi)+e_2(\varphi)+e_3(\varphi)+e_4(\varphi)=0\,.
\end{equation}
Indeed, formula \eqref{definition of strain formula} implies that condition \eqref{incompressibility 1} is equivalent to $\det(S+\mathrm{Id})=1$. The elementary identity from linear algebra
\begin{equation*}
\det(S+\mathrm{Id})=1+\operatorname{tr}S+\frac12
\left[
(\operatorname{tr}S)^2-\operatorname{tr}(S^2)
\right]
+
\operatorname{tr}\operatorname{adj}S
+
\det S
\end{equation*}
and \eqref{scalar invariant 1}--\eqref{scalar invariant 4} then give us \eqref{incompressibility 2}.

Formula \eqref{incompressibility 2}
allows us to express one of the four scalar invariants via the other three.
It is convenient to express $e_1$ via $e_2$, $e_3$ and $e_4$.
Then our action \eqref{action 1} takes the form
\begin{equation}
\label{action 2}
J(\varphi)=
\int_ML\bigl(e_2(\varphi),e_3(\varphi),e_4(\varphi)\bigr)
\,\rho(x)
\,dx\,,
\end{equation}
where
\begin{equation}
\label{relation between two Ls}
L(e_2,e_3,e_4)=\mathcal{L}(-e_2-e_3-e_4,e_2,e_3,e_4).
\end{equation}

Our mathematical model is formulated as follows:
vary the action \eqref{action 2} over volume preserving diffeomorphisms
$\mathrm{Diff}_\rho(M)$
and seek critical points.
The $L$ appearing in formula \eqref{action 2}
is some prescribed smooth real-valued function of three real variables
which characterises the physical properties
of our 4-dimensional isotropic homogeneous continuum.

We shall now impose two conditions
on the choice of the Lagrangian $L$.

\begin{itemize}
\item[]
\textbf{Condition 1\ }
We assume that
\begin{subequations}
\begin{equation}
\label{condition 1 formula initial}
\left.
\frac{\partial L}{\partial e_2}
\right|_{e_2=e_3=e_4=0}
\ne0,
\end{equation}
which is the minimal non-degeneracy condition.
This will be required
in Section~\ref{Linearised field equations}
where we will show that in a Ricci-flat spacetime our linearised field equations reduce
to Maxwell's equations.
Without loss of generality we assume further on that
\begin{equation}
\label{condition 1 formula 1}
\left.
\frac{\partial L}{\partial e_2}
\right|_{e_2=e_3=e_4=0}
=-1,
\end{equation}
\end{subequations}
which can always be achieved by rescaling.

\item[]
\textbf{Condition 2\ }
We assume that the function of one variable
$L(e_2,0,0)$ has a critical point on the positive real axis:
\begin{equation}
\label{condition 2 formula 1}
\left.
\frac{\partial L}{\partial e_2}
\right|_{e_2=c,\ e_3=e_4=0}
=0
\quad\text{for some}\quad c>0.
\end{equation}
This will be required
in Section~\ref{Explicit massive solutions of nonlinear field equations}
where we will construct explicit massive solutions
of our nonlinear field equations in Minkowski spacetime.
\end{itemize}

\begin{example}
[Examples
\ref{example of Linear Lagrangian}
and
\ref{example of Quadratic Lagrangian}
continued]
For the Lagrangian
\eqref{relation between two Ls}, \eqref{harmonic map action}
we get precisely \eqref{condition 1 formula 1},
whereas for the Lagrangian
\eqref{relation between two Ls}, \eqref{L mathcal quadratic}
we get
\[
\left.
\frac{\partial L}{\partial e_2}
\right|_{e_2=e_3=e_4=0}
=\beta,
\]
so condition
\eqref{condition 1 formula initial} is satisfied
when we have \eqref{beta nonzero}.

As to condition
\eqref{condition 2 formula 1},
it is not satisfied for the Lagrangian
\eqref{relation between two Ls}, \eqref{harmonic map action},
whereas for the Lagrangian
\eqref{relation between two Ls}, \eqref{L mathcal quadratic}
it is satisfied if and only if $\alpha\beta<0$.
\end{example}

\section{Nonlinear field equations}
\label{Nonlinear field equations}

Recall that the action in our mathematical model is defined by formula
\eqref{action 2}. Our field equations are obtained by varying
this action with respect to the unknown diffeomorphism $\varphi$
subject to the volume preservation constraint
\eqref{incompressibility 1}.

In order to write down the field equations let us initially
disregard the constraint \eqref{incompressibility 1}
and argue along the lines of \cite[Chapter 8]{jost}.
In local coordinates our action \eqref{action 2} can be written as
\begin{equation}
\label{temp1}
J(\varphi)=
\int
f
\left(
x^\alpha,\varphi^\beta,\frac{\partial\varphi^\gamma}{\partial x^\kappa}
\right)
\rho(x)\,
dx\,,
\end{equation}
where $\varphi^\beta$ is
the local representation \eqref{local representation of diffeomorphism}
of our diffeomorphism
and $f$ is some function of $x$, $\varphi(x)$ and the first partial derivatives
of $\varphi(x)$.
We vary $\varphi(x)$ as
\begin{equation}
\label{temp2}
\varphi^\beta(x)\mapsto\varphi^\beta(x)+\Delta\varphi^\beta(x),
\end{equation}
where $\Delta\varphi^\beta(x)$ is a small
smooth perturbation with small compact support.
Standard variational arguments
involving integration by parts
give us the variation of \eqref{temp1} in the form
\begin{equation}
\label{temp3}
\Delta J(\varphi)=
\int
E_\lambda
\left(
x^\alpha,\varphi^\beta,\frac{\partial\varphi^\gamma}{\partial x^\kappa}\,,
\frac{\partial^2\varphi^\sigma}{\partial x^\mu\partial x^\nu}
\right)
\Delta\varphi^\lambda
\,\rho(x)\,
\,dx\,.
\end{equation}

The quantity $E_\lambda$ appearing in the RHS of \eqref{temp3} is a two-point tensor:
it behaves as a scalar under changes of local coordinates $x$
and as a covector under changes of local coordinates $y$.
Hence,
\begin{equation}
\label{temp4}
\varphi\mapsto
E_\lambda
\left(
x^\alpha,\varphi^\beta,\frac{\partial\varphi^\gamma}{\partial x^\kappa}\,,
\frac{\partial^2\varphi^\sigma}{\partial x^\mu\partial x^\nu}
\right)
\end{equation}
is an invariantly defined map from diffeomorphisms to
covector fields.

We write the RHS of \eqref{temp4} in concise form as $E(\varphi)$. Thus, the field equations for the unconstrained
action \eqref{action 2} read
\begin{equation}
\label{temp5}
E(\varphi)=0.
\end{equation}
This is a system of four nonlinear second order partial differential equations for four unknowns,
the functions $\varphi^\alpha(x)$, $\alpha=1,2,3,4$, appearing in
the local representation \eqref{local representation of diffeomorphism}
of our diffeo\-morphism $\varphi$.

A detailed algorithm for the construction of the nonlinear differential operator $E$
is provided in
Appendix~\ref{Explicit field equations of nonlinear elasticity}.
However, 
we do not need the explicit form of $E$ for our purposes.
Even when we will be writing particular solutions of our nonlinear field equations,
see Sections
\ref{Explicit massless solutions of nonlinear field equations}
and
\ref{Explicit massive solutions of nonlinear field equations},
we will do this without using the explicit form of the operator $E$.

\begin{remark}
Straightforward analysis shows that the identity map is a solution
of \eqref{temp5}. Furthermore, any isometry from
$(M,g)$ to itself is a solution.
\end{remark}

Let us now incorporate the volume preservation constraint \eqref{incompressibility 1}
by adding to our original action \eqref{action 2} the term
\begin{equation}
\label{temp6}
K(\varphi,p):=
\int
\bigl[p(\varphi(x))\bigr]\,\bigl[\rho_\varphi(x)-\rho(x)\bigr]
\,dx\,,
\end{equation}
where
$\,\rho_\varphi(x):=
\sqrt{-\det h_{\alpha\beta}(x)}\,$
and $p:M\to\mathbb{R}$ is an additional unknown scalar
function playing the role of a Lagrange multiplier.
The function $p$ can be interpreted as pressure, cf.~\cite{ozanski}.

We will now vary our diffeomorphism as in \eqref{temp2}.

\begin{lemma}
\label{lemma about pressure}
The formula for the variation of the functional \eqref{temp6} reads
\begin{equation}
\label{variation of the volume preservation functional}
\Delta K(\varphi,p)=
-
\int
\left[
\frac{\partial p}{\partial y^\alpha}
(\varphi(x))
\right]
\bigl[
\Delta\varphi^\alpha(x)
\bigr]
\,\bigl[\rho(x)\bigr]
\,dx\,.
\end{equation}
\end{lemma}

\begin{proof}
Observe that the diffeomorphism $\varphi$ appears in formula \eqref{temp6} twice,
so
\begin{equation}
\label{sum of two variations}
\Delta K(\varphi,p)
=
\Delta K_1(\varphi,p)
+
\Delta K_2(\varphi,p),
\end{equation}
where
\begin{equation}
\label{temp6 subscript 1}
K_1(\varphi,p):=
\int
\bigl[p(\varphi(x))\bigr]\,\bigl[\rho_{\boldsymbol\varphi}(x)-\rho(x)\bigr]
\,dx\,,
\end{equation}
\begin{equation}
\label{temp6 subscript 2}
K_2(\varphi,p):=
\int
\bigl[p(\boldsymbol\varphi(x))\bigr]\,\bigl[\rho_\varphi(x)\bigr]
\,dx\,,
\end{equation}
the bold script indicating that this particular occurrence of $\varphi$ is not
subject to variation \eqref{temp2}.

Variation of \eqref{temp6 subscript 1} gives us
\begin{equation}
\label{temp6 subscript 1 variation 1}
\Delta K_1(\varphi,p)=
\int
\left[
\frac{\partial p}{\partial y^\alpha}
(\boldsymbol\varphi(x))
\right]
\bigl[
\Delta\varphi^\alpha(x)
\bigr]
\,\bigl[\rho_{\boldsymbol\varphi}(x)-\rho(x)\bigr]
\,dx\,.
\end{equation}

In order to calculate the variation of \eqref{temp6 subscript 2} we switch from local coordinates $x$
to local coordinates $y$ in accordance with
$\,y=\boldsymbol\varphi(x)\,$.
Formula \eqref{temp6 subscript 2} now reads
\begin{equation}
\label{temp6 subscript 2 variation 1}
K_2(\varphi,p)=
\int
\bigl[p(y)\bigr]\,\bigl[\mu_\varphi(y)\bigr]
\,dy^1dy^2dy^3dy^4,
\end{equation}
where $\mu_\varphi$ is the representation of the density $\rho_\varphi$
in local coordinates $y$.
A straightforward calculation,
see also
\eqref{incompressibility condition linearised}
and
\eqref{without covariant derivatives 2},
shows that
\begin{equation}
\label{temp6 subscript 2 variation 2}
\Delta\mu_\varphi(y)=
\frac
{
\partial
\bigl(
\bigl[\mu_{\boldsymbol\varphi}(y)\bigr]
\bigl[
\Delta\varphi^\alpha(\boldsymbol\varphi^{-1}(y))
\bigr]
\bigr)
}
{\partial y^\alpha}\,.
\end{equation}
{
Indeed, Jacobi's formula tells us that for any invertible matrix $A$ we have
\begin{equation}
\label{Jacobi formula}
\Delta(\det A)= (\det A)\operatorname{tr} \left(A^{-1}\, \Delta A\right).
\end{equation}
Applying \eqref{Jacobi formula} to the density $\rho_\varphi(x)$ we obtain
\begin{equation}
\label{variation of density basic}
\Delta \rho_\varphi(x)=\frac12 \rho_{\boldsymbol\varphi}(x)\, h^{\alpha\beta}(x)\, \Delta h_{\alpha\beta}(x).
\end{equation}
The variation of \eqref{definition of perturbed metric} reads
\begin{equation}
\label{variation of h}
\Delta h_{\alpha\beta}(x)=
\frac{\partial g_{\mu\nu}}{\partial y^\gamma}(\boldsymbol\varphi(x)) \Delta \varphi^\gamma
\frac{\partial\boldsymbol\varphi^\mu}{\partial x^\alpha}
\frac{\partial\boldsymbol\varphi^\nu}{\partial x^\beta}+
g_{\mu\nu}(\boldsymbol\varphi(x)) 
\frac{\partial\Delta\varphi^\mu}{\partial x^\alpha}
\frac{\partial\boldsymbol\varphi^\nu}{\partial x^\beta}
+
g_{\mu\nu}(\boldsymbol\varphi(x)) 
\frac{\partial\boldsymbol\varphi^\mu}{\partial x^\alpha}
\frac{\partial\Delta\varphi^\nu}{\partial x^\beta}.
\end{equation}
Formula \eqref{definition of perturbed metric} implies
that the inverse metric tensor $h^{\alpha\beta}$ can be written as
\begin{equation}
\label{h inverse in terms of psi}
h^{\alpha\beta}(x)=g^{\mu\nu}(\varphi(x))\, \psi_\mu{}^\alpha\,\psi_\nu{}^\beta,
\end{equation}
where the two-point tensor $\psi_\alpha{}^\beta$ is defined by the identity
\begin{equation}
\label{identity for psi}
\psi_\alpha{}^\beta(x)
\ \frac{\partial\varphi^\gamma}{\partial x^\beta}
(x)
=\delta_\alpha{}^\gamma.
\end{equation}
Substituting  \eqref{variation of h} and \eqref{h inverse in terms of psi} into \eqref{variation of density basic}
and using \eqref{identity for psi}, we get
\[
\Delta \rho_\varphi(x)
=
\rho_{\boldsymbol\varphi}(x)
\left[
\frac12
g^{\mu\nu}(\boldsymbol\varphi(x))
\frac{\partial g_{\mu\nu}}{\partial y^\gamma}(\boldsymbol\varphi(x))
\Delta \varphi^\gamma
+
\frac{\partial\Delta\varphi^\mu}{\partial y^\mu}(x)
\right].
\]
The above formula and the fact that
\[
\frac{\Delta\mu_\varphi}{\mu_{\boldsymbol\varphi}}(\boldsymbol\varphi(x))
=
\frac{\Delta\rho_\varphi}{\rho_{\boldsymbol\varphi}}(x)
\]
imply \eqref{temp6 subscript 2 variation 2}.
}

Substituting
\eqref{temp6 subscript 2 variation 2}
into
\eqref{temp6 subscript 2 variation 1}
and integrating by parts, we get
\begin{equation}
\label{temp6 subscript 2 variation 3}
\Delta K_2(\varphi,p)=
-
\int
\left[
\frac{\partial p}{\partial y^\alpha}
(y)
\right]
\bigl[
\Delta\varphi^\alpha(\boldsymbol\varphi^{-1}(y))
\bigr]
\,\bigl[\mu_{\boldsymbol\varphi}(y)\bigr]
\,dy^1dy^2dy^3dy^4.
\end{equation}
It remains only to switch back
from local coordinates $y$
to local coordinates $x$.
Formula \eqref{temp6 subscript 2 variation 3} becomes
\begin{equation}
\label{temp6 subscript 2 variation 4}
\Delta K_2(\varphi,p)=
-
\int
\left[
\frac{\partial p}{\partial y^\alpha}
(\boldsymbol\varphi(x))
\right]
\bigl[
\Delta\varphi^\alpha(x)
\bigr]
\,\bigl[\rho_{\boldsymbol\varphi}(x)\bigr]
\,dx\,.
\end{equation}

Substituting
\eqref{temp6 subscript 1 variation 1}
and
\eqref{temp6 subscript 2 variation 4}
into
\eqref{sum of two variations}
we arrive at
\eqref{variation of the volume preservation functional}.
\end{proof}

Lemma \ref{lemma about pressure} tells us that
the field equations for the constrained
action \eqref{action 2} read
\begin{equation}
\label{temp5 constrained}
E(\varphi)-\mathrm{d}p=0,
\end{equation}
where $\mathrm{d}p$ is the gradient of pressure $p$.
Equations \eqref{temp5 constrained} have to be supplemented by
the volume preservation condition \eqref{incompressibility 1}.

The term $\mathrm{d}p$ appearing in formula \eqref{temp5 constrained}
can be written in local coordinates as
\begin{equation}
\label{dp with Jacobian}
(\mathrm{d}p)_\alpha(x)=\psi_\alpha{}^\beta(x)
\ \frac{\partial(p\circ\varphi)}{\partial x^\beta}(x)\,,
\end{equation}
where the two-point tensor $\psi_\alpha{}^\beta$ was defined in \eqref{identity for psi}.

Formulae \eqref{temp5 constrained} and \eqref{incompressibility 1}
give us a system of five partial differential equations for five unknowns,
the functions $\varphi^\alpha(x)$, $\alpha=1,2,3,4$, appearing in
the local representation \eqref{local representation of diffeomorphism}
of our diffeomorphism $\varphi$ and the scalar field $(p\circ\varphi)(x)$.

\section{Displacements and rotations}
\label{Displacements and rotations}

Suppose that our diffeomorphism $\varphi:M\to M$
is sufficiently close to the identity map.
Then it can be described by a vector field of displacements $A$.
This vector field can be equivalently defined in two different ways.

Take an arbitrary point $P\in M$ and let $\Omega\subset M$
be a normal, with respect to $g$, neighbourhood of $P$.
As $\varphi$ is close to the identity map we can assume, without loss of generality,
that $\varphi(P)\in\Omega$.
Let $\gamma:[0,1]\to\Omega$ be the unique shortest geodesic, with respect to $g$,
connecting $P$ and $\varphi(P)$, so that $\gamma(0)=P$ and $\gamma(1)=\varphi(P)$.
Furthermore, let us parameterize our geodesic in such a way that
$\gamma(\tau)$ is a solution of the equation
\begin{equation*}
\ddot\gamma^\lambda+
\Gamma^\lambda{}_{\mu\nu}
\dot\gamma^\mu\dot\gamma^\nu=0,
\end{equation*}
where the $\Gamma^\lambda{}_{\mu\nu}$ are Christoffel symbols
and the dot stands for differentiation in $\tau$.
Then
\begin{equation}
\label{definition of displacement 1}
A(P):=\dot\gamma(0).
\end{equation}

Alternatively, let $W(P,Q)$ be the Ruse--Synge world function
\cite[Chapter II, \textsection 1]{synge} with respect to $g$.
Here $P,Q\in M$ are assumed to be sufficiently close.
Let $W'(P,Q):=\left. \operatorname{grad}_xW(x,Q)\right|_{x=P}$
be the gradient of the world function with respect to the first variable.
Then
\begin{equation}
\label{definition of displacement 2}
A^\flat(P):=-W'(P,\varphi(P)).
\end{equation}

In formula \eqref{definition of displacement 1} $A$ is a vector,
whereas in formula \eqref{definition of displacement 2} $A^\flat$ is a covector.
Raising and lowering tensor indices via the metric $g$
turns one into the other, see Appendix \ref{Exterior calculus notation}
for notation.

Working with a vector field of displacements $A$
rather than an abstract diffeomorphism $\varphi$
makes the physical interpretation clearer.

The field of displacements generates rotations. Describing these rotations mathematically
is the subject of finite strain theory in continuum mechanics
\cite[Section 23]{truesdell}. In what follows we present this construction
in a version adapted to Lorentzian signature and curved spacetime.

Consider the quantity
\begin{equation}
\label{deformation gradient 1}
\frac{\partial\varphi^\nu}{\partial x^\beta}(x)\,.
\end{equation}
The quantity \eqref{deformation gradient 1} is a two-point tensor:
it transforms as a covector under changes of local coordinates $x$
and as a vector under changes of local coordinates $y$.
The two-point tensor \eqref{deformation gradient 1}
describes a linear map from $T_PM$ to $T_{\varphi(P)}M$,
\[
v^\alpha\mapsto \frac{\partial\varphi^\nu}{\partial x^\beta}\,v^\beta.
\]

Let us now parallel transport \eqref{deformation gradient 1},
in the upper tensor index and with respect to the Levi-Civita connection associated with $g$,
along the unique shortest
geodesic from $\varphi(P)$ to $P$.
This gives us a (one-point) (1,1)-tensor $D^\nu{}_\beta(x)$
known in continuum mechanics as the \emph{deformation gradient}.
In local coordinates, the deformation gradient is explicitly obtained as follows. Let us choose the same coordinates $x$ and $y$ and let
\[
\gamma:[0,1]\to M, \quad \gamma(0)=\varphi(P), \quad \gamma(1)=P,
\]
be the unique geodesic connecting $\varphi(P)$ and $P$.
This is, of course, the same geodesic $\gamma$ as in the beginning of this section,
only now we use a different parameterization, $\tau\mapsto s:=1-\tau$.
Then $D^\nu{}_\beta:=X^\nu{}_\beta(1)$, where $X^\nu{}_{\beta}(s)$ is the unique solution to the ordinary differential equation
\begin{equation}
\label{parallel transport equation}
\dfrac{dX^\nu{}_\beta(s)}{ds}+X^\alpha{}_\beta(s)\,\Gamma^\nu{}_{\alpha\mu}(\gamma(s)) \dfrac{d\gamma^\mu(s)}{ds}=0
\end{equation}
subject to the initial condition $X^\nu{}_\beta(0)=\frac{\partial\varphi^\nu}{\partial x^\beta}$.
The deformation gradient describes, pointwise, a non-degenerate linear map in the fibres of the tangent bundle,
\begin{equation}
\label{deformation gradient as a linear map}
v^\alpha\mapsto D^\nu{}_\beta\,v^\beta.
\end{equation}
Moreover, formula \eqref{definition of perturbed metric}
can now be rewritten as
\begin{equation}
\label{perturbed metric via deformation gradient}
h_{\alpha\beta}(x)=
[D^\mu{}_\alpha(x)]\,[g_{\mu\nu}(x)]\,[D^\nu{}_\beta(x)]\,.
\end{equation}

\begin{remark}
Even though, in general, there is no explicit closed formula relating the vector field of displacements $A$ and the deformation gradient $D$, there is a connection between the two expressed by formula \eqref{deformation gradient linearised}. In special circumstances such a connection can be made sharper, for instance when working in Minkowski space, see formula \eqref{deformation gradient in Minkowski space}.
\end{remark}

Further on we assume that the linear map
\eqref{deformation gradient as a linear map}
is sufficiently close to the identity.
The issue at hand is to decompose
\eqref{deformation gradient as a linear map}
into a composition of a stretching map and a rotation map.
This is achieved by means of the polar decomposition.
The concept of polar decomposition is standard in linear algebra, 
only now it has to be adapted to Lorentzian signature.
Some work in this direction was done in \cite{bolshakov,mehl}.

\begin{definition}
We call a linear map $v^\alpha\mapsto B^\alpha{}_\beta\,v^\beta$
\emph{Lorentz--symmetric} if
$g_{\alpha\gamma}B^\gamma{}_\beta=g_{\beta\gamma}B^\gamma{}_\alpha$,
\emph{Lorentz--antisymmetric} if
$g_{\alpha\gamma}B^\gamma{}_\beta=-g_{\beta\gamma}B^\gamma{}_\alpha$
and
\emph{Lorentz--orthogonal} if
$B^\mu{}_\alpha\,g_{\mu\nu}\,B^\nu{}_\beta=g_{\alpha\beta}\,$.
\end{definition}

Any linear map
\eqref{deformation gradient as a linear map}
sufficiently close to the identity can be uniquely decomposed as
\begin{equation}
\label{polar decomposition}
D^\alpha{}_\beta=U^\alpha{}_\gamma\,V^\gamma{}_\beta\,,
\end{equation}
where $U$ is Lorentz--orthogonal
and $V$ is Lorentz--symmetric and close to the identity.

Recall that the standard polar decomposition theorem tells us that an invertible real matrix $K$ can be written as the product of an orthogonal matrix $U_K$ and a symmetric matrix $V_K$, $K=U_KV_K$, where these matrices are given by the explicit formulae
\[
V_K=(K^TK)^{1/2} \quad \text{and}\quad U_K=K (K^TK)^{-1/2}.
\]
The existence of polar decomposition
\eqref{polar decomposition}
can be established, for example, 
by using the power series
expansion for the function $\sqrt{1+z}$
with
$\,z=g^{\alpha\gamma}\,D^\mu{}_\gamma\,g_{\mu\nu}\,D^\nu{}_\beta
-\delta^\alpha{}_\beta\,$ and arguing along the lines of the classical proof.
Indeed, the quantity $D^\mu{}_\gamma\,g_{\mu\nu}\,D^\nu{}_\beta$ plays now the role of $K^TK$ --- see also \eqref{perturbed metric via deformation gradient} --- and the fact that \eqref{deformation gradient as a linear map} is sufficiently close to the identity ensures that the series converges and that the outcome is invertible.

In the setting of classical elasticity theory (Riemannian signature)
the tensor $V$ appearing in formula \eqref{polar decomposition}
is called the \emph{right stretch tensor}, see \cite[p.~53]{truesdell}.

Formula \eqref{polar decomposition}
and the fact that $D$ and $V$ are close to the identity
imply that $U$ is close to the identity as well.
Therefore, $U$ can be uniquely represented as
\begin{equation}
\label{U via F}
U=e^F,
\end{equation}
where $F$ is Lorentz--antisymmetric and small.
The tensor $F$ can be recovered from the tensor $U$ by using
the power series expansion for the function $\ln(1+z)$
with $z=U^\alpha{}_\beta-\delta^\alpha{}_\beta$.

Applying the above procedure to the deformation gradient we arrive at
a Lorentz--antisymmetric (1,1)-tensor $F^\alpha{}_\beta(x)$.
Lowering the first tensor index via $g$,
we get a covariant antisymmetric tensor $F_{\alpha\beta}(x)$
which can be viewed as a 2-form. We call it the \emph{rotation 2-form}.

Substituting \eqref{polar decomposition}
into \eqref{perturbed metric via deformation gradient}
we get
\begin{equation}
\label{perturbed metric via V}
h_{\alpha\beta}(x)=
[V^\mu{}_\alpha(x)]\,[g_{\mu\nu}(x)]\,[V^\nu{}_\beta(x)]\,.
\end{equation}

\begin{remark}
The order of indices in our polar decomposition \eqref{polar decomposition}
is important. Had we done the polar decomposition the other way round,
i.e.~as $\,D^\alpha{}_\beta=V^\alpha{}_\gamma\,U^\gamma{}_\beta\,$,
we wouldn't have gotten~\eqref{perturbed metric via V}.
\end{remark}

\begin{remark}
The assumption that our map \eqref{deformation gradient as a linear map} is sufficiently close to the identity warrants further clarification. As explained earlier in this section, this assumption is sufficient to ensure the existence of a Lorentzian polar decomposition. However, it is in general not necessary, as one can see by considering
the linear map \eqref{deformation gradient as a linear map} with
\[
D^\nu{}_\beta
=
\begin{pmatrix}
9& 0 \\
0 &-99
\end{pmatrix}
\]
in 2-dimensional Krein space ($g_{\kappa\lambda}=\operatorname{diag}(1,-1)$),
which clearly admits a Lorentzian polar decomposition.
The issue at hand is that necessary \emph{and} sufficient conditions for the existence of polar decomposition of matrices in indefinite inner product spaces are not available in the literature. While this is a very interesting problem in linear algebra, addressing it in the current paper would lead us astray. Therefore, we decided to impose what we view as a reasonable sufficient condition, with the plan to deal with the general theory of Lorentzian polar decomposition elsewhere.
\end{remark}

Formula \eqref{perturbed metric via V} tells us that
rotations do not appear explicitly in our mathematical model.
In other words, the physics described by our action
\eqref{action 2}
does not feel rotations.
However, we will still have to consider rotations later on in the paper
because they do not have a
life of their own: rotations are generated by displacements,
cf.~Sections
\ref{Massless Dirac equation}
and
\ref{Massive Dirac equation}.

The following lemma provides a list of
formulae obtained by linearising in $A$.
Some of them will be used
in Section \ref{Linearised field equations}.

\begin{lemma}
\label{trivial lemma}
We have
\begin{subequations}
\begin{align}
\label{deformation gradient linearised}
&D_{\alpha\beta}=g_{\alpha\beta}+\nabla_\beta A_\alpha+O(|A|^2),
\\
\label{V linearised}
&V_{\alpha\beta}=g_{\alpha\beta}+\frac12\,(\nabla_\alpha A_\beta+\nabla_\beta A_\alpha)+O(|A|^2),
\\
\label{U linearised}
&U_{\alpha\beta}=g_{\alpha\beta}-\frac12\,(\nabla_\alpha A_\beta-\nabla_\beta A_\alpha)+O(|A|^2),\\
\label{F linearised}
&F_{\alpha\beta}
=-\frac12\,(\nabla_\alpha A_\beta-\nabla_\beta A_\alpha)+O(|A|^2),
\\
\label{strain linearised}
&S_{\alpha\beta}=\nabla_\alpha A_\beta+\nabla_\beta A_\alpha+O(|A|^2),
\\
\label{incompressibility condition linearised}
&\dfrac{\det h_{\kappa\lambda}}{\det g_{\mu\nu}}=1+2\,\nabla_\alpha A^\alpha+O(|A|^2).
\end{align}
\end{subequations}
\end{lemma}

In the above lemma
and further on $\nabla$ is the Levi-Civita connection associated with $g$
and tensor indices are raised and lowered using the metric $g$.
In particular, the tensor in the LHS of formula \eqref{strain linearised}
is our original strain tensor \eqref{definition of strain formula}
but with the first tensor index lowered.
Of course, we have $S_{\alpha\beta}=h_{\alpha\beta}-g_{\alpha\beta}$.

Note that formulae
\eqref{F linearised}
and
\eqref{incompressibility condition linearised}
can be equivalently rewritten without
covariant derivatives using the identities
\begin{equation}
\label{without covariant derivatives 1}
\nabla_\alpha A_\beta-\nabla_\beta A_\alpha
=\partial_\alpha A_\beta-\partial_\beta A_\alpha
=(\mathrm{d}A^\flat)_{\alpha\beta}\,,
\end{equation}
\begin{equation}
\label{without covariant derivatives 2}
\nabla_\alpha A^\alpha=\rho^{-1}\partial_\alpha(\rho A^\alpha)=-\delta A^\flat,
\end{equation}
where $\rho$ is our Lorentzian density \eqref{defiition of Lorentzian density}.
See Appendix \ref{Exterior calculus notation} for exterior calculus notation.

\begin{proof}[Proof of Lemma~\ref{trivial lemma}]
Working in the neighbourhood of a point $x$ and in chosen local coordinates, our diffeomorphism $\varphi$ can be written as
\begin{equation}
\label{proof trivial lemma equation 1}
\varphi^\alpha(x)=x^\alpha + A^\alpha(x)+O(|A|^2)
\end{equation}
(here we made use of \eqref{definition of displacement 1}),
so that we have
\begin{equation}
\label{proof trivial lemma equation 2}
\dfrac{\partial \varphi^\alpha}{\partial x^\beta}=\delta^\alpha{}_\beta+ \dfrac{\partial A^\alpha}{\partial x^\beta} +O(|A^2|).
\end{equation}
Solving \eqref{parallel transport equation} at first order in $A$ with initial condition \eqref{proof trivial lemma equation 2}
and using the fact that $\frac{d\gamma^\alpha}{ds}=-A^\alpha+O(|A|^2)$
we get \eqref{deformation gradient linearised}.

Formula \eqref{deformation gradient linearised} implies
\begin{equation*}
\label{proof trivial lemma equation 3}
g^{\alpha\gamma}\,D^\mu{}_\gamma\,g_{\mu\nu}\,D^\nu{}_\beta=\delta^\alpha{}_\beta+\nabla_\beta A^\alpha+\nabla^\alpha A_\beta+O(|A|^2),
\end{equation*}
which gives us \eqref{V linearised} upon extracting the square root.

Formula \eqref{U linearised} is now obtained by combining \eqref{deformation gradient linearised} and \eqref{V linearised} with \eqref{polar decomposition}; in turn, it implies \eqref{F linearised} upon extracting the logarithm.

Substituting \eqref{proof trivial lemma equation 2} into \eqref{definition of perturbed metric} and, in turn, \eqref{definition of perturbed metric} into \eqref{definition of strain formula}, we get \eqref{strain linearised}.

Finally, \eqref{strain linearised} and the identity
\[
\det h_{\alpha\beta}=\det g_{\alpha\beta} \det(\operatorname{Id}+S),
\]
cf.~\eqref{definition of strain formula}, give us \eqref{incompressibility condition linearised}.
\end{proof}

\begin{remark}
There is an alternative way of describing a diffeomorphism in terms of a vector field.
This alternative approach is in the spirit of fluid mechanics and is based
on Lie-algebraic considerations. Namely, consider a smooth vector field
$u^\alpha(x)$, a field of `velocities', and
the autonomous system of ordinary differential equations
\begin{equation}
\label{integral curve 1}
\begin{cases}
\dot{y}=u(y),\\
\left.y\right|_{\tau=0}=x,
\end{cases}
\end{equation}
that it generates.
Here $\tau\in[0,1]$ is a parameter and the dot stands for differentiation in $\tau$.
We denote the solution of \eqref{integral curve 1} by $y(\tau;x)$.
For $u$ small enough the map $x\mapsto y(1;x)$ realises a diffeomorphism close to the identity.
At a formal level one would hope to generate an arbitrary diffeomorphism close to the identity
by a suitable choice of vector field $u$.
Furthermore, if we choose a divergence-free vector field, i.e.~a vector field satisfying
$\rho^{-1}\partial_\alpha(\rho u^\alpha)=0$
(compare with
\eqref{incompressibility condition linearised}
and
\eqref{without covariant derivatives 2}),
then for $u$ small enough the map $x\mapsto y(1;x)$
realises a volume-preserving diffeomorphism close to the identity.
Unfortunately, this approach doesn't work:
it is known \cite[p.~163]{marsden}
that there does not exist a neighbourhood of the
identity where the exponential map $\exp:\mathrm{Vect}(M)\to\mathrm{Diff}(M)$,
from vector fields $u$ to diffeomorphisms,
is surjective.
There are simple explicit examples of diffeomorphisms
of $\mathbb{S}^1$
arbitrarily close to the identity
that cannot be represented in terms of the above flow, see, for example,
\cite[p.~1017]{milnor},
\cite[p.~8--9]{banyaga},
\cite[p.~456--457]{kriegl}.
The description of a diffeomorphism in terms of a
vector field of displacements $A$ (see beginning of this section)
does not suffer from the deficiencies of the fluid mechanics description
\eqref{integral curve 1}. The fundamental difference between the
two approaches is that the concept of displacement relies on the
use of the metric structure.
\end{remark}

\section{Linearised field equations}
\label{Linearised field equations}

Carrying on from Section \ref{Displacements and rotations},
we assume that our diffeomorphism $\varphi:M\to M$
is sufficiently close to the identity map,
so that it
can be described by a vector field of displacements $A(x)$.
Furthermore, we can choose the local coordinates $y$
to be the same as $x$. 
Our aim in the current section is to linearise the field equations
\eqref{temp5 constrained}, \eqref{incompressibility 1}
in $A(x)$ and $p(x)$.

Formulae
\eqref{incompressibility condition linearised}
and
\eqref{without covariant derivatives 2}
give us the linearisation of the volume preservation
condition~\eqref{incompressibility 1}:
\begin{equation}
\label{linearisation of the volume preservation condition}
\delta A^\flat=0.
\end{equation}

Formula \eqref{dp with Jacobian} now reads
\[
(\mathrm{d}p)_\alpha(x)=
\frac{\partial p}{\partial x^\alpha}(\varphi(x))\,,
\]
and its linearisation is the usual gradient
\[
\frac{\partial p}{\partial x^\alpha}(x)\,.
\]

The issue at hand is the linearisation of $E(\varphi)$.

Inspection of formulae
\eqref{scalar invariant 2}--\eqref{scalar invariant 4},
\eqref{condition 1 formula 1}
and
\eqref{strain linearised}
shows that the expansion of our Lagrangian
$L\bigl(e_2(A),e_3(A),e_4(A)\bigr)$ in terms homogeneous in $A$
starts with the quadratic expression
\begin{equation}
\label{quadratised L}
L^{(2)}(A)=
-2\,
(\nabla_\alpha A^\alpha)^2
+\frac12\,
(\nabla_\alpha A_\beta+\nabla_\beta A_\alpha)
(\nabla^\alpha A^\beta+\nabla^\beta A^\alpha)
\,,
\end{equation}
so that $L\bigl(e_2(A),e_3(A),e_4(A)\bigr)=L^{(2)}(A)+O(|A|^3)$.
Variation of the quadratic action
\begin{equation*}
J^{(2)}(A)=
\int_M
L^{(2)}(A)
\,\rho(x)
\,dx
\end{equation*}
generates the linearisation $E^{(1)}(A)$ of $E(\varphi)$:
\begin{equation*}
\Delta J^{(2)}(A)=
\int
E_\lambda^{(1)}(A)
\ \Delta A^\lambda\,
\,\rho(x)\,
\,dx\,.
\end{equation*}
However, prior to variation it is useful to rewrite
\eqref{quadratised L}
as in the following lemma, whose proof is a straightforward computation.

\begin{lemma}
The Lagrangian \eqref{quadratised L}
can be equivalently rewritten as
\begin{equation}
\label{lemma about appearance of Maxwell Lagrangian equation 1}
L^{(2)}(A)=
\frac12\,
(\nabla_\alpha A_\beta-\nabla_\beta A_\alpha)
(\nabla^\alpha A^\beta-\nabla^\beta A^\alpha)
-2\,\mathrm{Ric}_{\mu\nu}\,A^\mu A^\nu
+\nabla_\kappa B^\kappa
\,,
\end{equation}
where $\mathrm{Ric}$ is the Ricci tensor associated with $g$ and
\[
B^\kappa=
-2
\left[
A^\kappa(\nabla_\gamma A^\gamma)-A^\gamma(\nabla_\gamma A^\kappa)
\right].
\]
\end{lemma}

The divergence term $\nabla_\kappa B^\kappa$
in formula \eqref{lemma about appearance of Maxwell Lagrangian equation 1}
does not contribute to the
field equations, so we can replace our
Lagrangian \eqref{quadratised L} with
\begin{equation}
\label{quadratised L nice}
\tilde L^{(2)}(A)=
\|\mathrm{d}A^\flat\|_g^2-2\,\mathrm{Ric}(A,A),
\end{equation}
see Appendix \ref{Exterior calculus notation} for exterior calculus notation.
The advantage of writing our quadratic Lagrangian in the form
\eqref{quadratised L nice} is that this representation does not
involve covariant derivatives.

Formula \eqref{quadratised L nice} implies that
the linearised operator generated by our action
\eqref{action 2} reads
\begin{equation}
\label{quadratised L nice after variation}
E^{(1)}=2\delta\mathrm{d}-4\,\mathrm{Ric}.
\end{equation}

In formulae
\eqref{quadratised L nice}
and
\eqref{quadratised L nice after variation}
we abuse notation by using the symbol $\mathrm{Ric}$
for two different objects,
the quadratic form on vectors $\mathrm{Ric}(u,u):=\mathrm{Ric}_{\alpha\beta}u^\alpha u^\beta$
and the linear map on covectors
$\mathrm{Ric}:v_\alpha\mapsto\mathrm{Ric}_{\alpha\beta}v^\beta$.

Hence, our linearised field equations
\eqref{temp5 constrained}, \eqref{incompressibility 1}
read
\[
\begin{pmatrix}
\delta\mathrm{d}-2\,\mathrm{Ric}&-\frac12\mathrm{d}\\
\delta&0
\end{pmatrix}
\begin{pmatrix}
A^\flat\\
p
\end{pmatrix}
=0.
\]
If we introduce a new scalar field
\begin{equation}
\label{definition of p tilde}
\tilde p:=-\frac12p
\end{equation}
the above system takes the form
\begin{equation}
\label{Linearised field equations main formula}
\begin{pmatrix}
\delta\mathrm{d}-2\,\mathrm{Ric}&\mathrm{d}\\
\delta&0
\end{pmatrix}
\begin{pmatrix}
A^\flat\\
\tilde p
\end{pmatrix}
=0.
\end{equation}

Let us now briefly discuss the analytic properties of the $5\times5$ matrix linear partial differential operator
\begin{equation}
\label{Definition of operator Lin}
\operatorname{Lin}:
\Omega^1(M)\oplus\Omega^0(M)
\to
\Omega^1(M)\oplus\Omega^0(M),
\qquad
\begin{pmatrix}
v\\
f
\end{pmatrix}
\mapsto
\begin{pmatrix}
\delta\mathrm{d}-2\,\mathrm{Ric}&\mathrm{d}\\
\delta&0
\end{pmatrix}
\begin{pmatrix}
v\\
f
\end{pmatrix}.
\end{equation}

We start with the observation that the operator $\operatorname{Lin}$ is formally self-adjoint (symmetric) with respect
to the $L^2$ inner product defined as in
Appendix~\ref{Exterior calculus notation}.

The more specific properties of a linear differential operator
are determined by its principal symbol.
In local coordinates, the principal symbol is obtained
by leaving only the leading (higher order)
derivatives and replacing each partial differentiation
$\partial/\partial x^\alpha$ by $i\xi_\alpha$,
where $\xi$ is the dual variable (momentum), see
\cite[subsection 1.1.3]{mybook}.
This gives a (matrix-)function on the cotangent bundle the properties of which determine the basic features
of the differential operator such as ellipticity or hyperbolicity. However, for our operator $\operatorname{Lin}$
matters are slightly more complicated because it has a block structure
\[
\begin{pmatrix}
2^\text{nd}\ \text{order operator}&1^\text{st}\ \text{order operator}\\
1^\text{st}\ \text{order operator}&0\ \text{order operator}
\end{pmatrix}
\]
with operators of different order in different blocks.
Matrix operators with this particular structure are called Agmon--Douglis--Nirenberg type operators
\cite{Agmon-Douglis-Nirenberg original paper}.
Application of the Agmon--Douglis--Nirenberg construction
gives the principal symbol of
$\operatorname{Lin}$ as the linear map
\begin{equation}
\label{Principal symbol of operator Lin}
\begin{pmatrix}
v\\
f
\end{pmatrix}
\mapsto
\begin{pmatrix}
\|\xi\|_g^2\,v-\langle\xi,v\rangle_g\,\xi+if\xi\\
-i\,\langle\xi,v\rangle_g
\end{pmatrix}.
\end{equation}
The determinant of the linear map \eqref{Principal symbol of operator Lin} is
\begin{equation}
\label{power 8}
-\|\xi\|_g^8\,.
\end{equation}
Now, if our metric $g$ were Riemannian then the quantity
\eqref{power 8} would not vanish on $T^*M\setminus\{0\}$
and, hence, our operator $\operatorname{Lin}$
would be elliptic in the Agmon--Douglis--Nirenberg sense.
However, for Lorentzian metric $g$ the quantity \eqref{power 8} vanishes
on light cones, 
which suggests that our operator $\operatorname{Lin}$ is hyperbolic.
There is extensive literature dealing with Agmon--Douglis--Nirenberg type
operators in the elliptic
setting but we are unaware of similar results for the hyperbolic case.
A rigorous investigation of well-posedness issues for the operator
$\operatorname{Lin}$ is, though, outside the scope of our paper.
For a review of different notions of hyperbolicity in a setting
similar to ours see \cite[Section 4]{wong}.

Note that if we replace the $5\times5$ matrix operator \eqref{Definition of operator Lin}
with the $4\times4$ matrix operator $\delta\mathrm{d}$, then the principal symbol
will be a degenerate matrix whose
determinant is identically zero.

Let us now assume that our spacetime $(M,g)$ is Ricci-flat,
\begin{equation}
\label{Ricci-flat}
\mathrm{Ric}=0.
\end{equation}
Note that condition \eqref{Ricci-flat} is the accepted relativistic definition of vacuum.
Moreover, it is easy to see that if $(M,g)$ is Ricci-flat, then so is $(M,h)$.

Under condition \eqref{Ricci-flat} equation
\eqref{Linearised field equations main formula}
implies
\begin{equation*}
\delta\mathrm{d}\tilde p=\square_g\,\tilde p=0.
\end{equation*}
We see that we have a separate equation for the scalar field $\tilde p$, the wave equation.
This observation allows us to collect solutions of our system
\eqref{Linearised field equations main formula}
into equivalence classes corresponding to particular choices of $\tilde p$:
we say that two solutions,
$\begin{pmatrix}
A^\flat\\
\tilde p
\end{pmatrix}$
and
$\begin{pmatrix}
{A^\flat}'\\
\tilde p'
\end{pmatrix}$,
are equivalent if $\tilde p=\tilde p'$.

Let us now fix a particular solution $\tilde p$ of the wave equation and work within
the corresponding equivalence class. Then
the first four equations from
our system \eqref{Linearised field equations main formula}
can be rewritten as
\begin{equation*}
\delta\mathrm{d}A^\flat=J,
\end{equation*}
where $J:=-\mathrm{d}\tilde p$.
We have arrived at Maxwell's equations in the Lorenz gauge
\eqref{linearisation of the volume preservation condition}
and with exact current $J\in\mathrm{d}\Omega^0(M)$.
Recovering Maxwell's equations in the Lorenz gauge is not a factitious artefact of our theory, but, in a sense, a natural thing to have: this is what one obtains when looking at irreducible representations of the Poincar\'e group in the spirit of Wigner's classification, cfr.~\cite[Chapter 21]{barut}.

\section{Homogeneous diffeomorphisms}
\label{Homogeneous diffeomorphisms}

In the remainder of this paper we will construct explicit solutions of the nonlinear field
equations~\eqref{temp5 constrained}.
Namely, we will write down explicitly volume preserving diffeomorphisms $\varphi$ satisfying
\eqref{temp5 constrained} with $p=0$.
In other words, we will present volume preserving solutions of the unconstrained nonlinear
field equations \eqref{temp5}.

Seeking such solutions constitutes an overdetermined problem: we are looking at a system of
five nonlinear partial differential equations \eqref{temp5}, \eqref{incompressibility 1}
for four unknowns,
the functions $\varphi^\alpha(x)$, $\alpha=1,2,3,4$, appearing in
the local representation \eqref{local representation of diffeomorphism}
of our diffeomorphism $\varphi$.
We will base our construction on group-theoretic ideas, the essence of which is explained
below.

Further on $\mathrm{Isom}(M,g)$ denotes the finite-dimensional subgroup of $\mathrm{Diff}(M)$
comprising diffeomorphisms that are isometries.

\begin{definition}\label{Definition of homogeneous diffeo}
Let $\varphi\in\mathrm{Diff}(M)$.
We say that $\varphi$ is \emph{homogeneous} if there exists
a subgroup $H\subset\mathrm{Isom}(M,g)$
acting transitively on $M$ and satisfying
\begin{equation}
\label{Homogeneous diffeomorphisms equation 1}
H\circ\varphi=\varphi\circ H.
\end{equation}
If we have the stronger property
\begin{equation}
\xi\circ\varphi=\varphi\circ\xi,
\qquad\forall\xi\in H,
\end{equation}
we say that $\varphi$ is \emph{equivariant}.
\end{definition}

In other words, condition \eqref{Homogeneous diffeomorphisms equation 1}
can be rewritten as follows: for any
$\xi\in H$ there exists a $\eta\in H$ such that the diagram
\begin{equation*}
\SelectTips{cm}{}
\xymatrix{
 M \ar[d]_-{\varphi} \ar[r]^-{\xi}& M \ar[d]^-{\varphi}\\
 M \ar[r]_-{\eta}  & M
}
\end{equation*}
is commutative.

\begin{theorem}
\label{theorem about homogeneous diffeomorphisms}
Let $\varphi$ be a homogeneous diffeomorphism.
Then the scalar invariants \eqref{scalar invariants} are constant.
Furthermore, if the covector field $E(\varphi)$
defined in accordance with formula \eqref{temp4}
vanishes at a point then it vanishes identically.
\end{theorem}

\begin{proof}
Let us prove the second statement first.
Let $\varphi$ be a homogeneous diffeomorphism and $x,y\in M$ two arbitrary points. We will assume that $\left. E(\varphi)\right|_{x}=0$ and we will show that $\left. E(\varphi)\right|_{y}=0$. In view of Definition \ref{Definition of homogeneous diffeo}, there exist isometries
$\xi$ and $\eta$ such that
\begin{equation}
\label{Matteo 1}
y=\xi(x), \qquad \varphi(y)=\eta(\varphi(x))\,,
\end{equation}
and
\begin{equation}
\label{Matteo 2}
\eta\circ\varphi =\varphi\circ \xi.
\end{equation}
 Note that in writing \eqref{Matteo 1} we only used the transitivity
of the action of $H$ on $M$,
whereas \eqref{Matteo 2} required the use of the additional condition
\eqref{Homogeneous diffeomorphisms equation 1}.

It is possible to choose coordinates in some neighbourhoods $\mathcal{U}(x)$ and $\mathcal{U}(\varphi(x))$ of $x$ and $\varphi(x)$ respectively in such a way that $\varphi$ is locally the identity map:
\[
\varphi|_{\mathcal{U}(x)}\simeq \mathrm{id}: \mathcal{U}(x)\to \mathcal{U}(\varphi(x)).
\]
We can then prescribe coordinates in some neighbourhood $\mathcal{U}(y)$ of $y$ (resp.~$\mathcal{U}(\varphi(y))$ of $\varphi(y)$) via the isometry $\xi$ (resp.~$\eta$). This has two consequences. Firstly, the map 
\[
\varphi|_{\mathcal{U}(y)}: \mathcal{U}(y)\to \mathcal{U}(\varphi(y))
\]
is the identity in our local coordinates.
Secondly, in this coordinate representation the components of the metric tensor
are the same near $x$ and $y$ and near $\varphi(x)$ and $\varphi(y)$.
This can be easily seen by explicitly imposing the isometry conditions
$\xi^*g=g$ and $\eta^*g=g$
locally, after observing that $\left.\xi\right|_{\mathcal{U}(x)}\simeq \mathrm{id}$
and $\left.\eta\right|_{\mathcal{U}(\varphi(x))}\simeq \mathrm{id}$ for our choice of coordinates. In particular, the Jacobian of the change of coordinates from coordinates centred at $x$ (resp.~$\varphi(x)$) to coordinates centred at $y$ (resp.~$\varphi(y)$) is $1$. The local expression \eqref{temp4} of $E(\varphi)$  depends only on the (local representation of the) metric, $\varphi$ and its derivatives. Since such local representations are the same in neighbourhoods of $x$ and $y$, $E(\varphi)|_x=0$ implies $E(\varphi)|_{y}=0$.

Finally, let us prove that the scalar invariants are constant. If we compute the scalar invariants in local coordinates, we realise that they only depend on the local representation of the metric, of $\varphi$ and of its first derivatives, see \eqref{definition of strain formula} and \eqref{scalar invariants}. Since such representations can be made the same in the neighbourhood of any pair of points $x$ and $y$, as described above, it ensues that the scalar invariants take the same value everywhere, namely, they are constant.
\end{proof}

Theorem \ref{theorem about homogeneous diffeomorphisms}
tells us that if we seek a solution of nonlinear field equations \eqref{temp5}
in the form of a homogeneous diffeomorphism then it is sufficient to satisfy
these field equations at a single point.

\begin{remark}
Note that our mathematical model is invariant under the action of the group of isometries in the following sense.
Let $\varphi\in\mathrm{Diff}_\rho(M)$ and $p:M\to\mathbb{R}$  be a solution of
our field equations \eqref{temp5 constrained},
and let $\xi\in\mathrm{Isom}(M,g)$ be an arbitrary isometry.
Then $\xi\circ\varphi$ and $p\circ\xi^{-1}$ is also a solution.
\end{remark}

\section{Special subgroups of the Poincar\'e group}
\label{Special subgroups of the Poincare group}

In the remainder of this paper we work in Minkowski space $\mathbb{M}$ where
the metric is $g_{\alpha\beta}=\operatorname{diag}(1,1,1,-1)$.
Further on $\mathrm{Poinc}(\mathbb{M}):=\mathrm{Isom}(\mathbb{R}^4,g)$
denotes the 10-dimensional group of isometries of $\mathbb{M}$,
commonly known as the Poincar\'e group.
Clearly, $\mathrm{Poinc}(\mathbb{M})=\mathbb{R}^4\rtimes\mathrm{O}(3,1)$.

In fact, we will be working with the identity component of the Poincar\'e group,
$\mathrm{ISO}^+(3,1)$. This is known to be the fundamental symmetry group of physics, in that it turns inertial frames into one another.

The Poincar\'e group can be realised as a subgroup of the matrix group
$\mathrm{SL}(5,\mathbb{R})$ as follows:
\[
\mathbb{R}^4\rtimes\mathrm{O}(3,1) \ni
(v,\Lambda)
\mapsto
\begin{pmatrix}
\Lambda & v\\
\begin{array}{cccc}0 & 0 & 0 & 0\end{array}&1
\end{pmatrix}
\in\mathrm{SL}(5,\mathbb{R}).
\]
Here the $5\times5$ matrix acts on $x\in\mathbb{M}$
by matrix vector multiplication
after complementing it with 1,
\[
\begin{pmatrix}
x\\
1
\end{pmatrix}
\mapsto
\begin{pmatrix}
\Lambda & v\\
\begin{array}{cccc}0 & 0 & 0 & 0\end{array}&1
\end{pmatrix}
\begin{pmatrix}
x\\
1
\end{pmatrix}.
\]

We now introduce special subgroups of the restricted Poincar\'e group
$\mathrm{ISO}^+(3,1)$
which will be used later in
Sections
\ref{Explicit massless solutions of nonlinear field equations}
and
\ref{Explicit massive solutions of nonlinear field equations}. 

\begin{definition}
\label{massless screw groups}
The
\emph{right-handed massless screw group} $\mathrm{SG}_0^+$
and
\emph{left-handed massless screw group} $\mathrm{SG}_0^-$
are the subgroups of $\mathrm{ISO}^+(3,1)$
realised in matrix representation by
\begin{equation}
\label{massless screw group equation 1}
\mathrm{SG}_0^\pm:=
\left\{
{
\left.
\begin{pmatrix}
\cos(q^3+q^4)&\mp\sin(q^3+q^4)&0&0&q^1\\
\pm\sin(q^3+q^4)&\cos(q^3+q^4)&0&0&q^2\\
0&0&1&0&q^3\\
0&0&0&1&q^4\\
0&0&0&0&1
\end{pmatrix}
\quad
\right|
\quad q\in\mathbb{R}^4
}
\right\}.
\end{equation}
\end{definition}

\begin{definition}
Let $m$ be a positive real number.
The \emph{massive screw group} $\mathrm{SG}_m$
is the subgroup of $\mathrm{ISO}^+(3,1)$
realised in matrix representation by
\begin{equation}
\label{massive screw group equation 1}
\mathrm{SG}_m:=
\left\{
{
\left.
\begin{pmatrix}
\cos(2mq^4)&-\sin(2mq^4)&0&0&q^1\\
\sin(2mq^4)&\cos(2mq^4)&0&0&q^2\\
0&0&1&0&q^3\\
0&0&0&1&q^4\\
0&0&0&0&1
\end{pmatrix}
\quad
\right|
\quad q\in\mathbb{R}^4
}
\right\}.
\end{equation}
\end{definition}

It is easy to see that
$\mathrm{SG}_0^+$,
$\mathrm{SG}_0^-$
and
$\mathrm{SG}_m$
are indeed subgroups of
$\mathrm{ISO}^+(3,1)$
and act transitively on $\mathbb{M}$.
Each of these groups is isomorphic to the direct product of $\mathbb{R}$
with a 3-dimensional group of Bianchi type $\mathrm{Bi}(\mathrm{VII}_0)$.

\

Let $\xi\in\mathrm{ISO}^+(3,1)$.
Then
$\,\xi^{-1}\,\mathrm{SG}_0^+\,\xi\,$,
$\,\xi^{-1}\,\mathrm{SG}_0^-\,\xi\,$
and
$\,\xi^{-1}\,\mathrm{SG}_m\,\xi\,$
are also subgroups of
$\mathrm{ISO}^+(3,1)$.
The question we want to address is what happens under conjugation.

\begin{lemma}
\label{Lemma about difference under conjugation of massless screw groups}
There does not exist a
$\xi\in\mathrm{ISO}^+(3,1)$ such that
$\,\xi^{-1}\,\mathrm{SG}_0^+\,\xi\,=\,\mathrm{SG}_0^-\,$.
\end{lemma}
\begin{proof}
The result follows from Lemma~\ref{lemma about torsion}:
the Hodge dual of axial torsion associated with the two groups
lies on opposite sides of the light cone and conjugation by an element of
$\mathrm{ISO}^+(3,1)$ cannot change this.
\end{proof}

Lemma~\ref{Lemma about difference under conjugation of massless screw groups}
tells us that the groups $\mathrm{SG}_0^+$ and  $\mathrm{SG}_0^-$ are genuinely different,
in that one cannot be turned into the other by conjugation.

Let us now examine what happens when we conjugate the massive screw group.
It turns out that the situation here is completely different.
Namely, choose $\xi=\operatorname{diag}(1,-1,-1,1,1)$.
Then
\begin{multline*}
\label{massive screw group equation 2}
\xi^{-1}\,\mathrm{SG}_m\,\xi=
\left\{
{
\left.
\begin{pmatrix}
\cos(2mq^4)&\sin(2mq^4)&0&0&q^1\\
-\sin(2mq^4)&\cos(2mq^4)&0&0&-q^2\\
0&0&1&0&-q^3\\
0&0&0&1&q^4\\
0&0&0&0&1
\end{pmatrix}
\quad
\right|
\quad q\in\mathbb{R}^4
}
\right\}
\\
=
\left\{
{
\left.
\begin{pmatrix}
\cos(2mq^4)&\sin(2mq^4)&0&0&q^1\\
-\sin(2mq^4)&\cos(2mq^4)&0&0&q^2\\
0&0&1&0&q^3\\
0&0&0&1&q^4\\
0&0&0&0&1
\end{pmatrix}
\quad
\right|
\quad q\in\mathbb{R}^4
}
\right\}.
\end{multline*}
This means that a different choice of signs in
\eqref{massive screw group equation 1}
does not yield
a different family of subgroups.
The argument presented in this paragraph is in agreement with
Lemma~\ref{lemma about torsion}:
the Hodge dual of axial torsion associated with the massive group
is spacelike and conjugation moves this covector
without encountering obstructions.

\section{Explicit massless solutions of nonlinear field equations}
\label{Explicit massless solutions of nonlinear field equations}

Working in Minkowski space $\mathbb{M}$,
we will describe our diffeomorphism $\varphi$ by a vector field of displacements
\begin{equation}
\label{diffeormorphism in terms of A}
\varphi:x^\alpha\mapsto x^\alpha+A^\alpha(x).
\end{equation}
The concept of a vector field of displacements was introduced in
Section~\ref{Displacements and rotations}.
The special feature of  Minkowski space is that
we do not need to assume that our diffeomorphism is sufficiently close to the identity map.
The only restriction on the choice of vector field $A$ is
\begin{equation}
\label{restricion on A}
\det(D^\alpha{}_\beta)\ne0,
\end{equation}
where
\begin{equation}
\label{deformation gradient in Minkowski space}
D^\alpha{}_\beta=\delta^\alpha{}_\beta+\partial A^\alpha/\partial x^\beta
\end{equation}
is the deformation gradient, see formula \eqref{deformation gradient as a linear map}
and associated discussion.
Condition \eqref{restricion on A} ensures that we do indeed have a diffeomorphism,
a smooth invertible map.

We seek volume preserving solutions.
Examination of formula \eqref{perturbed metric via deformation gradient}
shows that in Minkowski space the volume preservation condition
\eqref{incompressibility 1}
reduces to $\left|\det(D^\alpha{}_\beta)\right|=1$,
which means that we either have
\begin{subequations}
\begin{equation}
\label{determinant of deformation gradient plus 1}
\det(D^\alpha{}_\beta)=+1
\end{equation}
or
\begin{equation}
\label{determinant of deformation gradient minus 1}
\det(D^\alpha{}_\beta)=-1.
\end{equation}
\end{subequations}
Solutions presented in this section and the next one will possess the property
\eqref{determinant of deformation gradient plus 1}.

We say that a real lightlike covector $p=(p_1,p_2,p_3,p_4)$
lies on the forward light cone if $p_4>0$.
We say that a complex vector $u=(u^1,u^2,u^3,u^4)$ is isotropic if
$u_\alpha\bar u^\alpha>0$
and
$u_\alpha u^\alpha=0$.

The use of the term `isotropic' is motivated by Cartan who used it in the 3-dimensional Euclidean setting.
If we choose a coordinate system such that $u^4=0$
our definition is equivalent to that in
\cite[Chapter III, Section I]{cartan}.

\begin{theorem}
\label{theorem explicit massless}
Let $p$ be a real lightlike covector on the forward light cone,
let $u$ be a complex isotropic vector orthogonal to $p$
and let
\begin{equation}
\label{A mathbb for massless solution}
\mathbb{A}^\alpha(x)=
u^\alpha\,e^{ip_\beta x^\beta}.
\end{equation}
Then the diffeomorphism \eqref{diffeormorphism in terms of A} with
\begin{equation}
\label{A for massless solution}
A(x)=\operatorname{Re}
\left[
\mathbb{A}(x)
\right]
\end{equation}
is volume preserving and satisfies the 
nonlinear field equations \eqref{temp5}.
\end{theorem}

\begin{proof}
We can perform a (unique) proper orthochronous Lorentz transformation of coordinates
so that formula
\eqref{A mathbb for massless solution}
reads
\begin{equation}
\label{A mathbb for massless solution simplified}
\mathbb{A}^\alpha(x)=
a
\begin{pmatrix}
1\\
\mp i\\
0\\
0
\end{pmatrix}
e^{i(x^3+x^4)},
\end{equation}
where $a=\sqrt{u_\alpha\bar u^\alpha/2}\,$.
Then \eqref{A for massless solution} becomes
\begin{equation}
\label{A for massless solution simplified}
A^\alpha(x)=
a
\begin{pmatrix}
\cos(x^3+x^4)\\
\pm\sin(x^3+x^4)\\
0\\
0
\end{pmatrix}.
\end{equation}
Substituting
\eqref{A for massless solution simplified}
into
\eqref{deformation gradient in Minkowski space}
we get the following explicit formula for the deformation gradient:
\begin{equation}
\label{deformation gradient for massless solution}
D^\alpha{}_\beta
=
\begin{pmatrix}
1&0&-a\sin(x^3+x^4)&-a\sin(x^3+x^4)\\
0&1&\pm a\cos(x^3+x^4)&\pm a\cos(x^3+x^4)\\
0&0&1&0\\
0&0&0&1
\end{pmatrix},
\end{equation}
where the first tensor index, $\alpha$, enumerates the rows
and the second, $\beta$, the columns.
It is immediately clear that \eqref{determinant of deformation gradient plus 1}
is satisfied.
Substituting now
\eqref{deformation gradient for massless solution}
into
\eqref{perturbed metric via deformation gradient}
and
\eqref{definition of strain formula}
we get the following explicit formula for the strain tensor:
\begin{equation}
\label{strain tensor for massless solution}
S^\alpha{}_\beta
=
\begin{pmatrix}
0&0&-a\sin(x^3+x^4)&-a\sin(x^3+x^4)\\
0&0&\pm a\cos(x^3+x^4)&\pm a\cos(x^3+x^4)\\
-a\sin(x^3+x^4)&\pm a\cos(x^3+x^4)&a^2&a^2\\
a\sin(x^3+x^4)&\mp a\cos(x^3+x^4)&-a^2&-a^2
\end{pmatrix}.
\end{equation}
It is easy to check that the matrix
\eqref{strain tensor for massless solution}
is nilpotent, so
all our scalar invariants \eqref{scalar invariants}
vanish identically.
Note that the nilpotency index of \eqref{strain tensor for massless solution}
is three, which, according to Lemma~\ref{Lemma about nilpotency index},
is the maximal possible.

We vary the vector field of displacements $A(x)$ as
\begin{equation*}
A^\alpha(x)\mapsto
A^\alpha(x)+\Delta A^\alpha(x).
\end{equation*}
This generates an increment of our scalar invariants
$\Delta e_j$
and an increment of our Lagrangian
\begin{equation*}
\sum_{j=2}^4
\left.
\frac{\partial L}{\partial e_j}
\right|_{e_2=e_3=e_4=0}
\Delta e_j\,.
\end{equation*}
In order to prove that our diffeomorphism
satisfies the 
nonlinear field equations \eqref{temp5}
it is sufficient to prove that
\begin{equation}
\label{sufficient massless}
\int_{\mathbb{R}^4} \Delta e_j
\,dx\,=0\,,
\qquad j=2,3,4.
\end{equation}

Straightforward calculations give
\begin{subequations}
\label{variation of scalar invariants massless}
\begin{align}
\label{variation of scalar invariants massless 1}
\Delta e_1&
=
2
\left(
\delta^\beta{}_\alpha
+
\frac{\partial A_\alpha}{\partial x_\beta}
\right)
\frac{\partial\Delta A^\alpha}{\partial x^\beta}\,,
\\
\label{variation of scalar invariants massless 2}
\Delta e_2&=-2
\left(
a^2\,p^\beta\,p_\alpha
+
\frac{\partial A^\beta}{\partial x^\alpha}
+
\frac{\partial A_\alpha}{\partial x_\beta}
\right)
\left(\frac{\partial\Delta A^\alpha}{\partial x^\beta}
+
\frac{\partial A_\gamma}{\partial x_\alpha}
\,
\frac{\partial\Delta A^\gamma}{\partial x^\beta}\right),
\\
\label{variation of scalar invariants massless 3}
\Delta e_3&=
2\,a^2\,p^\beta\,p_\alpha
\left(\frac{\partial\Delta A^\alpha}{\partial x^\beta}
+
\frac{\partial A_\gamma}{\partial x_\alpha}
\,
\frac{\partial\Delta A^\gamma}{\partial x^\beta}\right)
=
2\,a^2\,p^\beta\,p_\alpha
\frac{\partial\Delta A^\alpha}{\partial x^\beta}\,,
\\
\label{variation of scalar invariants massless 4}
\Delta e_4&=0\,,
\end{align}
\end{subequations}
where $p_\kappa=(0,0,1,1)$.
Integrating
\eqref{variation of scalar invariants massless 2}--\eqref{variation of scalar invariants massless 4}
by parts and using the identities
\[
\square A=0,
\qquad
\frac{\partial A^\alpha}{\partial x^\alpha}=0,
\qquad
\left(p^\alpha\frac{\partial}{\partial x^\alpha}\right)A=0,
\]
we arrive at \eqref{sufficient massless}.
\end{proof}

The crucial element of the above proof is the observation
that the scalar invariants \eqref{scalar invariants}
generated by the diffeomorphism
\eqref{diffeormorphism in terms of A},
\eqref{A for massless solution simplified}
are constant.
We established this fact by means of explicit analytic calculations.
However, at a group-theoretic level this follows from
Theorem~\ref{theorem about homogeneous diffeomorphisms}.
Indeed, take an arbitrary $\xi\in\mathrm{SG}_0^\pm$,
see formula
\eqref{massless screw group equation 1}.
This isometry acts as
\[
\xi:
\begin{pmatrix}
x^1\\
x^2\\
x^3\\
x^4\\
\end{pmatrix}
\mapsto
\begin{pmatrix}
x^1\cos(q^3+q^4)\mp x^2\sin(q^3+q^4)\\
\pm x^1\sin(q^3+q^4)+x^2\cos(q^3+q^4)\\
x^3\\
x^4
\end{pmatrix}
+
\begin{pmatrix}
q^1\\
q^2\\
q^3\\
q^4\\
\end{pmatrix}.
\]
Our diffeomorphism
\eqref{diffeormorphism in terms of A},
\eqref{A for massless solution simplified}
acts as
\[
\varphi_\pm:
\begin{pmatrix}
x^1\\
x^2\\
x^3\\
x^4\\
\end{pmatrix}
\mapsto
\begin{pmatrix}
x^1\\
x^2\\
x^3\\
x^4\\
\end{pmatrix}
+
a
\begin{pmatrix}
\cos(x^3+x^4)\\
\pm\sin(x^3+x^4)\\
0\\
0
\end{pmatrix}
\]
and its inverse acts as
\[
\varphi_\pm^{-1}:
\begin{pmatrix}
x^1\\
x^2\\
x^3\\
x^4\\
\end{pmatrix}
\mapsto
\begin{pmatrix}
x^1\\
x^2\\
x^3\\
x^4\\
\end{pmatrix}
-
a
\begin{pmatrix}
\cos(x^3+x^4)\\
\pm\sin(x^3+x^4)\\
0\\
0
\end{pmatrix}.
\]
Composing $\xi$ with $\varphi_\pm$ we get
\begin{multline*}
\xi\circ\varphi_\pm:
\begin{pmatrix}
x^1\\
x^2\\
x^3\\
x^4\\
\end{pmatrix}
\mapsto
\begin{pmatrix}
x^1\cos(q^3+q^4)\mp x^2\sin(q^3+q^4)\\
\pm x^1\sin(q^3+q^4)+x^2\cos(q^3+q^4)\\
x^3\\
x^4
\end{pmatrix}
+
\begin{pmatrix}
q^1\\
q^2\\
q^3\\
q^4\\
\end{pmatrix}
\\
+a
\begin{pmatrix}
\cos(x^3+q^3+x^4+q^4)\\
\pm\sin(x^3+q^3+x^4+q^4)\\
0\\
0
\end{pmatrix}.
\end{multline*}
Finally, a composition with $\varphi_\pm^{-1}$ gives us
\begin{multline*}
\varphi_\pm^{-1}\circ\xi\circ\varphi_\pm:
\begin{pmatrix}
x^1\\
x^2\\
x^3\\
x^4\\
\end{pmatrix}
\mapsto
\begin{pmatrix}
x^1\cos(q^3+q^4)\mp x^2\sin(q^3+q^4)\\
\pm x^1\sin(q^3+q^4)+x^2\cos(q^3+q^4)\\
x^3\\
x^4
\end{pmatrix}
+
\begin{pmatrix}
q^1\\
q^2\\
q^3\\
q^4\\
\end{pmatrix}
\\
+a
\begin{pmatrix}
\cos(x^3+q^3+x^4+q^4)\\
\pm\sin(x^3+q^3+x^4+q^4)\\
0\\
0
\end{pmatrix}
-a
\begin{pmatrix}
\cos(x^3+q^3+x^4+q^4)\\
\pm\sin(x^3+q^3+x^4+q^4)\\
0\\
0
\end{pmatrix},
\end{multline*}
which means that $\varphi_\pm^{-1}\circ\xi\circ\varphi_\pm=\xi$.
Thus, our diffeomorphism $\varphi_\pm$ is equivariant
as per Definition~\ref{Definition of homogeneous diffeo}
with $H=\mathrm{SG}_0^\pm$.

Observe now that the complex 2-form $p\wedge u^\flat$
is an eigenvector of the Hodge star.
This motivates the following definition.

\begin{definition}
\label{massless handedness}
We say that a solution from
Theorem \ref{theorem explicit massless}
is
\emph{right-handed} if $\,*(p\wedge u^\flat)=i\,(p\wedge u^\flat)\,$
and
\emph{left-handed} if $\,*(p\wedge u^\flat)=-i\,(p\wedge u^\flat)\,$.
\end{definition}

It is easy to see that the upper sign in formula
\eqref{A for massless solution simplified}
corresponds to a right-handed solution
and the lower sign corresponds to a left-handed one.
Note that we defined right/left-handedness for
groups (Definition \ref{massless screw groups})
and massless solutions (Definition \ref{massless handedness})
in such a way that they agree.

\section{Explicit massive solutions of nonlinear field equations}
\label{Explicit massive solutions of nonlinear field equations}

\begin{theorem}
\label{theorem explicit massive}
Let $m$ be a positive real number
and let $p$ be a real timelike covector with
$p_\beta p^\beta=-4m^2$ and $p_4>0$.
Let $u$ be a complex isotropic vector orthogonal to $p$,
and let $v$ be a real vector orthogonal to $p$ and $u$.
Suppose that
\begin{equation}
\label{strength of field condition 1}
4m^2\left(\frac12u_\alpha\bar u^\alpha+v_\beta v^\beta\right)=c\,,
\end{equation}
where $c$ is a critical point from \eqref{condition 2 formula 1},
and put
\begin{equation}
\label{A mathbb for massive solution}
\mathbb{A}^\alpha(x)=
u^\alpha\,e^{ip_\beta x^\beta}.
\end{equation}
Then the diffeomorphism \eqref{diffeormorphism in terms of A} with
\begin{equation}
\label{A for massive solution}
A(x)=\operatorname{Re}
\left[
\mathbb{A}(x)
\right]
+(p_\gamma x^\gamma)\,v
\end{equation}
is volume preserving and satisfies the
nonlinear field equations \eqref{temp5}.
\end{theorem}

\begin{remark}
It is easy to see that under the assumptions of
Theorem \ref{theorem explicit massive}
the scalar $\|\mathrm{d}A^\flat\|_g^2$ is constant,
\[
\|\mathrm{d}A^\flat\|_g^2=
-
4m^2\left(\frac12u_\alpha\bar u^\alpha+v_\beta v^\beta\right).
\]
Hence, formula \eqref{strength of field condition 1}
can be equivalently rewritten as
\begin{equation}
\label{strength of field condition 2}
\|\mathrm{d}A^\flat\|_g^2=-c\,,
\end{equation}
which is a condition on the strength of the field $\mathrm{d}A^\flat$.
We see a certain similarity with the Born--Infeld model \cite{born}, \cite[Section 2.1]{jimenez}
which sets constraints on admissible values of $\|\mathrm{d}A^\flat\|_g^2$.
\end{remark}

\begin{proof}[Proof of Theorem \ref{theorem explicit massive}]
Arguing as in the proof of Theorem \ref{theorem explicit massless},
we can perform a (unique) proper orthochronous Lorentz transformation of coordinates
so that formula \eqref{A mathbb for massive solution} reads
\begin{equation}
\label{A mathbb for massive solution simplified}
\mathbb{A}^\alpha(x)=
a
\begin{pmatrix}
1\\
-i\\
0\\
0
\end{pmatrix}
e^{2imx^4}
\end{equation}
and \eqref{A for massive solution} becomes
\begin{equation}
\label{A for massive solution simplified}
A^\alpha(x)=
\begin{pmatrix}
a\cos(2mx^4)\\
a\sin(2mx^4)\\
2mbx^4\\
0
\end{pmatrix}.
\end{equation}
Here
\begin{subequations}
\begin{equation}
\label{a massive}
a=\sqrt{\frac{u_\alpha\bar u^\alpha}2}\,,
\end{equation}
\begin{equation}
\label{b massive}
b=-\frac i{4ma^2}*\!(p\wedge u^\flat\wedge\bar u^\flat\wedge v^\flat).
\end{equation}

Note that $|b|=\sqrt{v_\alpha v^\alpha}$.
However, in defining the scalar invariant
$b$ we used the seemingly more complicated formula
\eqref{b massive} in order to capture information
on the relative orientation of the four covectors
$p$, $\operatorname{Re}u^\flat$, $\operatorname{Im}u^\flat$ and $v^\flat$.
With this notation formula \eqref{strength of field condition 1}
can be rewritten as
\begin{equation}
4m^2(a^2+b^2)=c\,.
\end{equation}
\end{subequations}

The corresponding deformation gradient reads
\begin{equation}
\label{deformation gradient for massive solution}
D^\alpha{}_\beta
=
\begin{pmatrix}
1&0&0&-2ma\sin(2mx^4)\\
0&1&0&2ma\cos(2mx^4)\\
0&0&1&2mb\\
0&0&0&1
\end{pmatrix},
\end{equation}
for which \eqref{determinant of deformation gradient plus 1} is satisfied.
The resulting strain tensor is
\begin{equation}
\label{strain tensor for massive solution}
S^\alpha{}_\beta
=
\begin{pmatrix}
0&0&0&-2ma\sin(2mx^4)\\
0&0&0&2ma\cos(2mx^4)\\
0&0&0&2mb\\
2ma\sin(2mx^4)&-2ma\cos(2mx^4)&-2mb&-c
\end{pmatrix}.
\end{equation}

Unlike
\eqref{strain tensor for massless solution},
the matrix
\eqref{strain tensor for massive solution} is not nilpotent:
its eigenvalues are zero (algebraic and geometric multiplicity two)
and
\[
-\frac c2
\pm\frac{\sqrt{c(c-4)}}2\,.
\]
The matrix is diagonalisable if and only if $c\ne4$.

The fact that the eigenvalues of the strain tensor
\eqref{strain tensor for massive solution}
are constant implies that
all our scalar invariants \eqref{scalar invariants}
are constant:
\[
e_1=-c,
\qquad
e_2=c,
\qquad
e_3=e_4=0.
\]
Arguing as in the proof of Theorem \ref{theorem explicit massless},
we see that
in order to prove that our diffeomorphism
satisfies the 
nonlinear field equations \eqref{temp5}
it is sufficient to show,
in view of \eqref{condition 2 formula 1},
that
\begin{equation*}
\label{sufficient massive}
\int_{\mathbb{R}^4} \Delta e_j
\,dx\,=0\,,
\qquad j=3,4.
\end{equation*}

It is easy to see that $\Delta e_4=0$, which, in essence,
is to do with the fact that zero is a double eigenvalue of
\eqref{strain tensor for massive solution}.

The formula for $\Delta e_3$ reads
\begin{equation*}
\Delta e_3=
B^\beta{}_\alpha
\,\frac{\partial\Delta A^\alpha}{\partial x^\beta}\,,
\end{equation*}
where the $B^\beta{}_\alpha$ is some tensor.
The explicit formulae for the components of this tensor
are complicated, however for our purposes it suffices to observe that
$B^4{}_\alpha=0$ and that the remaining components depend only on the
coordinate $x^4$. Hence, integration by parts yields
\[
\int_{\mathbb{R}^4} \Delta e_3
\,dx\,
=-\int_{\mathbb{R}^4}
\left(
\frac{\partial B^\beta{}_\alpha}{\partial x^\beta}
\right)
\Delta A^\alpha
\,dx\,
=-\int_{\mathbb{R}^4}
\left(
\frac{\partial B^4{}_\alpha}{\partial x^4}
\right)
\Delta A^\alpha
\,dx\,
=\,0\,.
\]
\end{proof}

Group-theoretic arguments apply to the massive case as well.
Taking an arbitrary $\xi\in\mathrm{SG}_m$,
see formula
\eqref{massive screw group equation 1},
we get
$
\varphi^{-1}\circ\xi\circ\varphi
=
\eta
$,
where
\[
\mathrm{SG}_m
\ni
\eta:
\begin{pmatrix}
x^1\\
x^2\\
x^3\\
x^4\\
\end{pmatrix}
\mapsto
\begin{pmatrix}
x^1\cos(2mq^4)-x^2\sin(2mq^4))\\
x^1\sin(2mq^4)+x^2\cos(2mq^4)\\
x^3\\
x^4
\end{pmatrix}
+
\begin{pmatrix}
q^1\\
q^2\\
q^3-2mbq^4\\
q^4\\
\end{pmatrix}.
\]
This means that our diffeomorphism $\varphi$ is homogeneous
as per Definition~\ref{Definition of homogeneous diffeo}
with $H=\mathrm{SG}_m$.
It is equivariant if and only if $b=0$.

Let us discuss the continuum mechanics interpretation of
formula \eqref{A for massive solution simplified}.
We are looking at a translation
(rigid motion without rotation)
of 3-dimensional Euclidean
space which is a function of the time coordinate $x^4$.
Every point of 3-dimensional Euclidean space moves along a helix,
see Figure \ref{fig}(i) for $b>0$
and Figure \ref{fig}(ii) for $b<0$.

The parameter $b$ could be interpreted as electric charge.
Note that for given values of positive parameters $m$ and $a$
the parameter $b$ can take only two values,
\[
b=\pm\sqrt{\frac c{4m^2}-a^2}\,.
\]

\begin{figure}
\centering
\begin{subfigure}{0.48\textwidth}
\begin{tikzpicture}
\begin{axis} [
    view={0}{30},
    axis lines=none,
    ymin=-2,
    ymax=5,
    xmin=-2,
    xmax=2]
    \addplot3 [thick, black, domain=0*pi:1.73*pi, samples = 500, samples y=0] ({sin(deg(-x))}, {cos(deg(-x))}, {x});
    \addplot3 [thick, black, domain=1.83*pi:3.73*pi, samples = 500, samples y=0] ({sin(deg(-x))}, {cos(deg(-x))}, {x});
        \addplot3 [thick, black, domain=3.83*pi:5.73*pi, samples = 500, samples y=0] ({sin(deg(-x))}, {cos(deg(-x))}, {x});
           \addplot3 [thick, black, domain=5.83*pi:7.73*pi, samples = 500, samples y=0] ({sin(deg(-x))}, {cos(deg(-x))}, {x});
  \addplot3 [thick, ->, black, domain=7.83*pi:9.73*pi, samples = 500, samples y=0] ({sin(deg(-x))}, {cos(deg(-x))}, {x});
\end{axis}
\end{tikzpicture}
\caption{$b>0$}
\label{fig1}
\end{subfigure}
\begin{subfigure}{0.48\textwidth}
\begin{tikzpicture}
\begin{axis} [
    view={0}{30},
    axis lines=none,
    ymin=-2,
    ymax=5,
    xmin=-2,
    xmax=2]
    \addplot3 [thick, black, domain=0*pi:-2.18*pi, samples = 500, samples y=0] ({sin(deg(x))}, {cos(deg(x))}, {x});
    \addplot3 [thick, black, domain=-2.28*pi:-4.18*pi, samples = 500, samples y=0] ({sin(deg(x))}, {cos(deg(x))}, {x});
        \addplot3 [thick, black, domain=-4.28*pi:-6.18*pi, samples = 500, samples y=0] ({sin(deg(x))}, {cos(deg(x))}, {x});
           \addplot3 [thick, black, domain=-6.28*pi:-8.18*pi, samples = 500, samples y=0] ({sin(deg(x))}, {cos(deg(x))}, {x});
  \addplot3 [thick, ->, black, domain=-8.28*pi:-9.73*pi, samples = 500, samples y=0] ({sin(deg(x))}, {cos(deg(x))}, {x});
\end{axis}
\end{tikzpicture}
\caption{$b<0$}
\label{fig2}
\end{subfigure}
\caption{Massive solution}
\label{fig}
\end{figure}
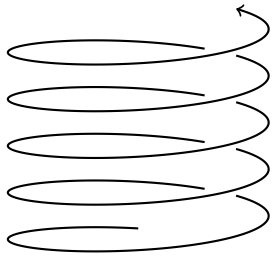
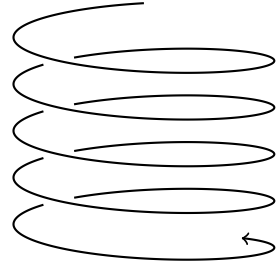

\section{Massless Dirac equation}
\label{Massless Dirac equation}

Let the diffeomorphisms $\varphi_+$ and $\varphi_-$
be right-handed and left-handed massless solutions
as per
Definition \ref{massless handedness}.
In this section we will calculate the corresponding rotation 2-forms,
see Section~\ref{Displacements and rotations},
and show that they are equivalent to spinor fields which satisfy
massless Dirac equations.

The deformation gradient reads
\begin{equation}
\label{deformation gradient massless abstract}
D^\alpha{}_\beta
=
\delta^\alpha{}_\beta
+
\operatorname{Re}
\left[
iu^\alpha p_\beta\,e^{ip_\gamma x^\gamma}
\right].
\end{equation}
In a particular coordinate system the above formula
turns to \eqref{deformation gradient for massless solution}.
Performing a polar decomposition
\eqref{polar decomposition},
we get
\begin{equation}
\label{tensor U massless abstract}
U^\alpha{}_\beta
=
\delta^\alpha{}_\beta
-
\frac12
\operatorname{Re}
\left[
i
\left(
p^\alpha u_\beta
-
u^\alpha p_\beta
\right)
e^{ip_\gamma x^\gamma}
\,
\right]
-
\frac{u_\gamma\bar u^\gamma}{16}p^\alpha p_\beta\,,
\end{equation}
\begin{equation*}
\label{tensor V massless abstract}
V^\alpha{}_\beta
=
\delta^\alpha{}_\beta
+
\frac12
\operatorname{Re}
\left[
i
\left(
p^\alpha u_\beta
+
u^\alpha p_\beta
\right)
e^{ip_\gamma x^\gamma}
\,
\right]
+
\frac{3u_\gamma\bar u^\gamma}{16}p^\alpha p_\beta\,.
\end{equation*}
On account of formula \eqref{U via F}
one can compute the logarithm of
\eqref{tensor U massless abstract},
lower the first index
and obtain
the following explicit formula for the rotation 2-form:
\begin{equation}
\label{tensor F massless abstract}
F
=
-\frac12
\operatorname{Re}
\left[
i
(p\wedge u^\flat)
e^{ip_\gamma x^\gamma}
\,
\right]
=
-
\frac12\mathrm{d}A^\flat.
\end{equation}

We see that the formula for our rotation 2-form is remarkably simple.
Recall that for a general diffeomorphism we have
$F=-\frac12\mathrm{d}A^\flat+O(|A|^2)$,
see formulae
\eqref{F linearised}
and
\eqref{without covariant derivatives 1}.
However the deformation gradient generated by our massless
solutions is very special and turns out to be linear in displacements,
without any second (or higher) order terms and without any assumptions
on the amplitude.
The underlying reason for such simplicity is that
at any given point of $\mathbb{M}$
one can identify a 2-dimensional invariant subspace
of the tangent fibre
in which
the deformation gradient \eqref{deformation gradient massless abstract}
differs from the identity map.
Furthermore, the restriction of the
Minkowski metric to this subspace is degenerate.

Put
\begin{equation}
\label{tensor F mathbb massless abstract}
\mathbb{F}
:=
-
\frac12\mathrm{d}\mathbb{A}^\flat
=-\frac i2(p\wedge u^\flat)\,e^{ip_\gamma x^\gamma},
\end{equation}
so that $F=\operatorname{Re}\mathbb{F}$.
In the remainder of this section we examine the structure of the
complex-valued 2-form $\mathbb{F}$.

The 2-form $\mathbb{F}$ is
polarised
\begin{equation}
\label{F mathbb massless polarised}
*\mathbb{F}=\pm i\,\mathbb{F}
\end{equation}
(cf.~Definition \ref{massless handedness})
and degenerate
\begin{equation*}
\det\mathbb{F}=0.
\end{equation*}
It is known, see Appendix \ref{Spinor representation of 2-forms},
that such a 2-form is equivalent, modulo sign,
to a spinor field which is, effectively, the square root of $\mathbb{F}$.
This spinor field is undotted, $\xi=\xi^a$, in the left-handed case
(lower sign in \eqref{F mathbb massless polarised})
and dotted, $\eta=\eta_{\dot a}\,$, in the right-handed case
(upper sign in~\eqref{F mathbb massless polarised}).

\begin{theorem}
\label{theorem Dirac massless}
The spinor field $\xi$ associated with a
left-handed massless solution
satisfies the massless Dirac equation
\begin{equation}
\label{massless Dirac left-handed}
\sigma^\alpha{}_{\dot ab}\,\partial_{x^\alpha}\xi^b=0.
\end{equation}
The spinor field $\eta$ associated with a
right-handed massless solution
satisfies the massless Dirac equation
\begin{equation}
\label{massless Dirac right-handed}
\sigma^{\alpha\dot ba}\,\partial_{x^\alpha}\eta_{\dot b}=0.
\end{equation}
\end{theorem}

\begin{proof}
It is sufficient to
establish the identities
\eqref{massless Dirac left-handed}
and
\eqref{massless Dirac right-handed}
in one coordinate system, so let us work in the coordinate system in
which we have \eqref{A mathbb for massless solution simplified}.
Plugging
\eqref{A mathbb for massless solution simplified}
into
\eqref{tensor F mathbb massless abstract}
we get
\begin{equation*}
\mathbb{F}_{\alpha\beta}
=
-\frac{ia}2
\begin{pmatrix}
0&0&-1&-1\\
0&0&\pm i&\pm i\\
1&\mp i&0&0\\
1&\mp i&0&0
\end{pmatrix}
e^{i(x^3+x^4)},
\end{equation*}
where the upper/lower sign corresponds to right-/left-handedness respectively.
Using formulae from Appendix \ref{Spinor representation of 2-forms}
we conclude that

\begin{equation}
\label{explicit xi massless}
\xi^a=
\pm
\sqrt{\frac a2}\,
\begin{pmatrix}
0\\
i
\end{pmatrix}
e^{i(x^3+x^4)/2},
\end{equation}
\begin{equation}
\label{explicit eta massless}
\eta_{\dot a}=
\pm
\sqrt{\frac a2}\,
\begin{pmatrix}
1\\
0
\end{pmatrix}
e^{i(x^3+x^4)/2}.
\end{equation}
It remains only to substitute
\eqref{Pauli matrices covariant}
and
\eqref{explicit xi massless}
into
\eqref{massless Dirac left-handed},
and
\eqref{Pauli matrices contravariant}
and
\eqref{explicit eta massless}
into
\eqref{massless Dirac right-handed}.
\end{proof}

\section{Massive Dirac equation}
\label{Massive Dirac equation}

Let the diffeomorphism $\varphi$
be a massive solution
as per
Theorem \ref{theorem explicit massive}.
The corresponding deformation gradient reads
\begin{equation}
\label{deformation gradient massive abstract}
D^\alpha{}_\beta
=
\delta^\alpha{}_\beta
+
\operatorname{Re}
\left[
iu^\alpha p_\beta\,e^{ip_\gamma x^\gamma}
\right]
+v^\alpha p_\beta.
\end{equation}
In a particular coordinate system the above formula
turns to
\eqref{deformation gradient for massive solution}.
Explicit calculations show that \eqref{deformation gradient massless abstract}
admits a polar decomposition if and only if $c<4$. Assuming that $c<4$
and arguing as in Section~\ref{Massless Dirac equation}
we arrive at
the following explicit formula for the rotation 2-form:
\begin{multline}
\label{tensor F massive abstract}
F
=
-
\,
\frac1{\sqrt c}
\,
\operatorname{arctanh}
\left(
\frac{\sqrt c}2
\,
\right)
\left(
\operatorname{Re}
\left[
i
(p\wedge u^\flat)
e^{ip_\gamma x^\gamma}
\,
\right]
+
(p\wedge v^\flat)
\right)
\\
=
-
\,
\frac1{\sqrt c}
\,
\operatorname{arctanh}
\left(
\frac{\sqrt c}2
\,
\right)
\mathrm{d}A^\flat.
\end{multline}
Observe that
unlike the massless case \eqref{tensor F massless abstract}
the prefactor in the RHS of \eqref{tensor F massive abstract}
brings about, effectively, contributions nonlinear in $A$,
see 
\eqref{strength of field condition 2}.
But apart from the prefactor formula \eqref{tensor F massive abstract} is quite simple.
Here the underlying reason is the same as in the massless case:
at any given point of $\mathbb{M}$
one can identify a 2-dimensional invariant subspace
of the tangent fibre
in which
the deformation gradient \eqref{deformation gradient massive abstract}
differs from the identity map.

Put
\begin{equation*}
\mathbb{F}
:=
-
\,
\frac1{\sqrt c}
\,
\operatorname{arctanh}
\left(
\frac{\sqrt c}2
\,
\right)
\mathrm{d}\mathbb{A}^\flat
=
-
\,
\frac i{\sqrt c}
\,
\operatorname{arctanh}
\left(
\frac{\sqrt c}2
\,
\right)
(p\wedge u^\flat)\,e^{ip_\gamma x^\gamma},
\end{equation*}
which captures information about the oscillating part of $F$.
As in the previous section, we will now
examine the geometric content of $\mathbb{F}$.

Unlike the massless case, $\mathbb{F}$ is not polarised. However,
it can be decomposed into a sum of polarised pieces
\begin{equation*}
\mathbb{F}=\mathbb{F}_++\mathbb{F}_-\,,
\end{equation*}
\begin{equation*}
\mathbb{F}_+
=
\frac{\mathbb{F}-i*\mathbb{F}}{2}\,,
\qquad
\mathbb{F}_-
=
\frac{\mathbb{F}+i*\mathbb{F}}{2}\,,
\end{equation*}
\begin{equation}
\label{decomposition of F mathbb 3}
*\mathbb{F}_\pm
=
\pm i\mathbb{F}_\pm\,.
\end{equation}
In our case the two polarised pieces are degenerate, i.e.
\begin{equation}
\label{decomposition of F mathbb 4}
\det\mathbb{F}_\pm=0.
\end{equation}
The latter follows easily from the observation that the pair of identities
\eqref{decomposition of F mathbb 4} is equivalent to
\begin{equation*}
\det\mathbb{F}=0,
\qquad
\mathbb{F}_{\alpha\beta}\,\mathbb{F}^{\alpha\beta}=0.
\end{equation*}

The 2-form $\mathbb{F}_-$ is equivalent, modulo sign,
to an undotted spinor field $\xi=\xi^a$
and the 2-form $\mathbb{F}_+$ is equivalent, modulo sign,
to a dotted spinor field $\eta=\eta_{\dot a}\,$.
Since in our case the scalar $\xi^a\bar\eta_a$ is real and nonzero,
one can choose the relative sign of $\xi$ and $\eta$
so that $\xi^a\bar\eta_a>0$.
Thus, our complex-valued 2-form $\mathbb{F}$
is equivalent to a bispinor field $(\xi,\eta)$.
This bispinor field is defined uniquely up to sign and is,
effectively, the square root of $\mathbb{F}$.

\begin{theorem}
The bispinor field $(\xi,\eta)$ associated with a
massive solution
satisfies the massive Dirac equation
\begin{equation}
\label{massive Dirac equation}
-i\sigma^\alpha{}_{\dot ab}\,\partial_{x^\alpha}\xi^b=m\,\eta_{\dot a}\,,
\qquad
-i\sigma^{\alpha\dot ba}\,\partial_{x^\alpha}\eta_{\dot b}=m\,\xi^a\,.
\end{equation}
\end{theorem}

\begin{proof}
Arguing along the same lines as that of
Theorem \ref{theorem Dirac massless}, in the special coordinate system
in which we have \eqref{A mathbb for massive solution simplified} we get
\[
\xi^a=\eta_{\dot a}
=
\pm
\sqrt{
\frac{ma}{\sqrt c}
\,
\operatorname{arctanh}
\left(
\frac{\sqrt c}2
\,
\right)
}
\begin{pmatrix}
0\\
i
\end{pmatrix}
e^{imx^4}.
\]
The above bispinor field clearly satisfies
\eqref{massive Dirac equation}.
\end{proof}

\begin{remark}
In writing the massive Dirac equation \eqref{massive Dirac equation}
we adopted the spinor representation,
cf.~\cite[formula (20.2)]{berestetskii},
as opposed to the standard representation,
cf.~\cite[formulae (21.19), (21.17)]{berestetskii}.
\end{remark}

\section{Acknowledgements}

We are grateful to
Z.~Avetisyan,
C.~G.~B\"ohmer,
C.~Dappiaggi,
E.~R.~Johnson,
M.~Levitin,
N.~Saveliev
and E.~Shargorodsky
for valuable suggestions and stimulating discussions.



\begin{appendices}

\section{Notation and conventions}

\subsection{Exterior calculus}
\label{Exterior calculus notation}

In this paper we identify differential forms with covariant antisymmetric tensors.
Henceforth $M$ is a 4-manifold equipped with Lorentzian metric $g$
and Levi-Civita connection $\nabla$.

It is well known that the metric $g$ induces a canonical isomorphism
between the tangent bundle $TM$
and the contangent bundle $T^*M$, the so-called \emph{musical isomorphism}.
We denote it by $\flat:TM\to T^*M$ and its inverse by $\sharp:T^*M\to TM$.

Given a scalar field $f\in C^\infty(M)$, its
exterior derivative $\mathrm{d}f$ is defined as the gradient.
Given a 1-form $A\in\Omega^1(M)$, its
exterior derivative $\mathrm{d}A\in\Omega^2(M)$
is defined, componentwise, as
\[
(\mathrm{d}A)_{\alpha\beta}
=
\partial_{x^\alpha}A_\beta-\partial_{x^\beta}A_\alpha\,.
\]

Given a pair of rank $k$ covariant antisymmetric tensors $Q$ and $T$
we define their pointwise inner product as
\begin{equation*}
\langle Q,T\rangle_g
:=
\frac{1}{k!}\,
\overline{Q}_{\alpha_1\dots\alpha_k}
\,
T_{\beta_1\dots\beta_k}
\,
g^{\alpha_1\beta_1}\cdots g^{\alpha_k\beta_k}
\,,
\end{equation*}
and, accordingly,
\begin{equation*}
\|Q\|_g^2
:=
\langle Q,Q\rangle_g
\,.
\end{equation*}

We define the $L^2$ inner product
\begin{equation*}
(Q,T)_{L^2}
:=
\int
\langle Q,T\rangle_g
\,\sqrt{-\det g_{\mu\nu}}
\ dx\,.
\end{equation*}

Given $U\in\Omega^k(M)$ and $V\in\Omega^{k-1}(M)$
we define the action of the codifferential 
$\delta:\Omega^k(M)\to\Omega^{k-1}(M)$ in accordance with
\[
\langle U,\mathrm{d}V\rangle=\langle\delta U,V\rangle.
\]
In particular, when $A\in\Omega^1(M)$ and $F\in\Omega^2(M)$,
we get in local coordinates
\[
\delta A=-\nabla^\alpha A_\alpha,
\]
\[
(\delta F)_\alpha=\nabla^\beta F_{\alpha\beta}.
\]

For the sake of clarity, let us mention that the wedge product of 1-forms
reads
\[
(A\wedge B)_{\alpha\beta}=A_\alpha B_\beta-A_\beta B_\alpha\,.
\]

We define the action of the Hodge star on a rank $k$ antisymmetric tensor as
\begin{equation*}
(*Q)_{\mu_{k+1}\ldots\mu_4}\!:=\frac1{k!}\,
\sqrt{-\det g_{\alpha\beta}}
\ Q^{\mu_1\ldots\mu_k}\,\varepsilon_{\mu_1\ldots\mu_4}\,,
\end{equation*}
where $\varepsilon$ is the totally antisymmetric symbol,
$\varepsilon_{1234}:=+1$.

\subsection{Spinors}

In this appendix as well as in Appendix~\ref{Spinor representation of 2-forms}
we restrict ourselves to
the special case of Minkowski space~$\mathbb{M}$.
We work with 2-component Weyl spinors as opposed to 4-component Dirac spinors.
We recall below the basic ideas and conventions, referring the reader to
\cite[Section 18]{berestetskii}
and \cite[Section 1.2]{buchbinder}
for further details.

In line with
\cite{berestetskii,buchbinder}
we treat spinors as holonomic objects.
This approach simplifies analysis in the case of flat space
and is traditionally used in particle physics.

We adopt the following conventions.
\begin{itemize}
\item
`Metric' spinor:
\begin{equation*}
\epsilon_{ab}=\epsilon_{\dot a\dot b}=
\epsilon^{ab}=\epsilon^{\dot a\dot b}=
\begin{pmatrix}
0&1\\
-1&0
\end{pmatrix}.
\end{equation*}
\item
`Covariant', with respect to spinor indices, Pauli matrices:
\begin{equation}
\label{Pauli matrices covariant}
\sigma^1{}_{\dot ab}:=
\begin{pmatrix}
0&1\\
1&0
\end{pmatrix},
\quad
\sigma^2{}_{\dot ab}:=
\begin{pmatrix}
0&-i\\
i&0
\end{pmatrix},
\quad
\sigma^3{}_{\dot ab}:=
\begin{pmatrix}
1&0\\
0&-1
\end{pmatrix},
\quad
\sigma^4{}_{\dot ab}:=
\begin{pmatrix}
1&0\\
0&1
\end{pmatrix}.
\end{equation}
\item
`Contravariant', with respect to spinor indices, Pauli matrices:
\begin{equation}
\label{Pauli matrices contravariant}
\sigma^{1\dot ab}=
\begin{pmatrix}
0&-1\\
-1&0
\end{pmatrix},
\quad
\sigma^{2\dot ab}=
\begin{pmatrix}
0&-i\\
i&0
\end{pmatrix},
\quad
\sigma^{3\dot ab}=
\begin{pmatrix}
-1&0\\
0&1
\end{pmatrix},
\quad
\sigma^{4\dot ab}=
\begin{pmatrix}
1&0\\
0&1
\end{pmatrix}.
\end{equation}
\end{itemize}
Here
$\sigma^{\alpha\dot ab}=\epsilon^{\dot a\dot c}\epsilon^{bd}\sigma^\alpha{}_{\dot cd}$.

Pauli matrices satisfy the identities
\begin{subequations}
\begin{equation}
\label{Pauli identity 1}
\sigma^{\alpha\dot b a}
\,
\sigma^\beta{}_{\dot bc}
+
\sigma^{\beta\dot b a}
\,
\sigma^\alpha{}_{\dot bc}
=-2g^{\alpha\beta}\delta^a{}_c\,,
\end{equation}
\begin{equation}
\label{Pauli identity 2}
\sigma^\alpha{}_{\dot a b}\,\sigma^{\beta\dot c b}
+
\sigma^\beta{}_{\dot a b}\,\sigma^{\alpha\dot c b}
=-2g^{\alpha\beta}\delta_{\dot a}{}^{\dot c}\,.
\end{equation}
\end{subequations}

\subsection{Spinor representation of 2-forms}
\label{Spinor representation of 2-forms}

Let $\mathbb{F}_-$ and $\mathbb{F}_+$ be polarised complex 2-forms,
see \eqref{decomposition of F mathbb 3}.
Then
$\mathbb{F}_-$ is equivalent to a trace-free undotted rank two spinor
$\zeta^b{}_c\,$,
\begin{subequations}
\begin{equation}
\label{F minus to undotted}
(\mathbb{F}_-)^{\alpha\beta}=
-i
\sigma^\alpha{}_{\dot ab}\,\zeta^b{}_c\,\sigma^{\beta\dot ac}\,,
\end{equation}
and
$\mathbb{F}_+$ is equivalent to a trace-free dotted rank two spinor
$\theta_{\dot b}{}^{\dot c}\,$,
\begin{equation}
\label{F plus to dotted}
(\mathbb{F}_+)^{\alpha\beta}=
i
\sigma^{\alpha\dot ba}\,\theta_{\dot b}{}^{\dot c}\,\sigma^\beta{}_{\dot ca}\,.
\end{equation}
\end{subequations}
The identities
\eqref{Pauli identity 1}
and
\eqref{Pauli identity 2}
ensure that that the right-hand sides of
\eqref{F minus to undotted}
and
\eqref{F plus to dotted},
respectively,
are antisymmetric in $\alpha,\beta$.

\begin{fact}
\label{fact 1}
The following are equivalent.
\begin{enumerate}[(i)]
\item
$\det\mathbb{F}_-=0\,$.
\item
$\det\zeta=0\,$.
\item
There exists a rank one spinor $\xi^a$ such that
$\zeta^b{}_c=\xi^b\,\xi^d\,\epsilon_{dc}\,$.
\end{enumerate}
\end{fact}

\begin{fact}
\label{fact 2}
The following are equivalent.
\begin{enumerate}[(i)]
\item
$\det\mathbb{F}_+=0\,$.
\item
$\det\theta=0\,$.
\item
There exists a rank one spinor $\eta_{\dot a}$ such that
$\theta_{\dot b}{}^{\dot c}=\eta_{\dot b}\,\eta_{\dot d}\,\epsilon^{\dot d\dot c}\,$.
\end{enumerate}
\end{fact}

Facts \ref{fact 1} and \ref{fact 2} imply that a degenerate polarised
2-form is equivalent to the square of a rank 1 spinor. The latter is defined
uniquely up to sign.

The equivalence between (i) and (ii) in the above statements is
a straightforward consequence of
\eqref{F minus to undotted}
and
\eqref{F plus to dotted},
whereas (iii) is not so obvious.
The relevant arguments are presented in
Appendix~\ref{Nilpotent operators in a 2D symplectic space}.

\section{Some results in linear algebra}

\subsection{Linear algebra involving a pair of quadratic forms}
\label{Linear algebra involving a pair of quadratic forms}

Working in an $n$-dimensional real vector space $V$, consider a pair of
non-degenerate symmetric bilinear forms,
$g:V\times V\to\mathbb{R}$
and
$h:V\times V\to\mathbb{R}$.
These uniquely define an invertible linear operator $L:V\to V$ via the formula
\[
h(u,v)=g(Lu,v),
\qquad\forall u,v\in V.
\]

The eigenvalue problem for the operator $L$
\[
Lu=\lambda u
\]
can be equivalently reformulated in terms of bilinear forms
\[
h(u,v)=\lambda g(u,v),
\qquad\forall v\in V.
\]
The expression $h-\lambda g$ is called a \emph{linear pencil}
of symmetric bilinear forms.

It is well known
\cite[Section X.6]{gantmacher}  
that if at least one of the forms is sign definite,
then $L$ has real eigenvalues and is diagonalisable.
In this case the associated pencil is called \emph{regular}.

If neither $g$ nor $h$ is sign definite, then the operator $L$
may have complex eigenvalues and may not be diagonalisable. In particular,
the \emph{strain} operator
\[
S:=L-\mathrm{Id}
\]
may be nilpotent. This is a fundamental difference with the regular
(sign definite) case where the strain operator cannot be nilpotent.

We now address the question what is
the maximal nilpotency index of $S$.

\begin{lemma}
\label{Lemma about nilpotency index}
Suppose that $n\ge4$ and that both $g$ and $h$
have Lorentzian signature
\[
\underset{n-1}{\underbrace{+\,\cdots\,+}}\,-\,.
\]
Then the nilpotency index of $S$ is less than or equal to three.
\end{lemma}

\begin{proof}
Observe first that it is sufficient to prove the lemma in the complex setting,
where we can use \cite[Theorem 8.4.1]{gohberg}.
Examination of the latter shows that nilpotency index strictly greater than
four is not possible, whereas nilpotency index equal to
four is possible only if we have an invariant subspace in which our
operator has the structure \cite[formula (8.4.19)]{gohberg}.
But the matrix $N$ from \cite[formula (8.4.19)]{gohberg}
with $\lambda=0$ has nilpotency index at most three.
\end{proof}

\begin{remark}
Closer examination shows that in our setting
the structure \cite[formula (8.4.19)]{gohberg}
cannot be realised because the latter describes
an operator which is Lorentz--normal but not Lorentz--symmetric.
The only way the strain operator can get nilpotency index three
is when it has a Jordan block of the type \cite[formula (8.4.18)]{gohberg}
with $\lambda=r=0$.
As a final observation, let us point out that in dimensions $n=2$ and $n=3$
the maximal nilpotency indices two and three can actually be attained.
\end{remark}

\subsection{Nilpotent operators in a 2D symplectic space}
\label{Nilpotent operators in a 2D symplectic space}

\begin{lemma}
Let $V$ be a 2-dimensional complex vector space
equipped with a symplectic form $\omega$
and let $L:V\to V$ be a linear operator.
Then $L$ is nilpotent if and only if there exists a $u\in V$ such that
\begin{equation}
\label{Lemma about symplectic space equation 1}
Lv=u\,\omega(u,v),
\qquad
\forall v\in V.
\end{equation}
\end{lemma}

\begin{proof}
An operator of the form \eqref{Lemma about symplectic space equation 1}
is clearly nilpotent. So we only need to prove the converse statement.

Let $L$ be nilpotent. Choose a basis in $V$ so that the symplectic form reads
\begin{equation}
\label{Lemma about symplectic space equation 2}
\omega(v,w)=\varepsilon_{rs}\,v^rw^s,
\end{equation}
where $\varepsilon$ is the totally antisymmetric symbol, $\varepsilon_{12}=+1$.
The linear operator $L$ is represented in this basis by the matrix
\begin{equation}
\label{Lemma about symplectic space equation 3}
L^r{}_s=
\begin{pmatrix}
a&b\\
c&d
\end{pmatrix}.
\end{equation}
The nilpotency condition is equivalent to the trace and the determinant of
$L$ both being zero.
Hence, \eqref{Lemma about symplectic space equation 3}
can be rewritten as
\begin{equation}
\label{Lemma about symplectic space equation 4}
L^r{}_s=
\begin{pmatrix}
\sqrt{-bc}&b\\
c&-\sqrt{-bc}
\end{pmatrix}
\end{equation}
with appropriate choice of complex square root.
The matrix
\eqref{Lemma about symplectic space equation 4}
can be factorised as
\begin{equation}
\label{Lemma about symplectic space equation 5}
L^r{}_s
=
\begin{pmatrix}
\sqrt b\\
-\sqrt{-c}
\end{pmatrix}
\begin{pmatrix}
\sqrt b
&
-\sqrt{-c}
\end{pmatrix}
\begin{pmatrix}
0&1\\
-1&0
\end{pmatrix},
\end{equation}
where the square roots are chosen in such a way that
$\sqrt b\ \sqrt{-c}\,=\,\sqrt{-bc}\,$.
Formulae 
\eqref{Lemma about symplectic space equation 5}
and
\eqref{Lemma about symplectic space equation 2}
give us
\eqref{Lemma about symplectic space equation 1}
with
\[
u=
\begin{pmatrix}
\sqrt b\\
-\sqrt{-c}
\end{pmatrix}.
\]
\end{proof}

\section{Differential geometric characterisation of screw groups}

Let $\mathrm{SG}$ be one of the screw groups
$\mathrm{SG}_0^+$, $\mathrm{SG}_0^-$ or $\mathrm{SG}_m$
defined in Section~\ref{Special subgroups of the Poincare group}. In what follows, the (global) isomorphism $T\mathbb{M}\simeq \mathbb{M} \times \mathbb{M}$ will be tacitly understood. In particular, we will not distinguish between points of $M$ and vectors in the tangent fibres.

Direct inspection shows that for any $P,Q\in\mathbb{M}$
there exists a unique $\xi\in\mathrm{SG}$
such that $\xi(P)=Q$. This allows us to define a map
\begin{eqnarray*}
\Upsilon: T_P\mathbb{M} \to  T_Q\mathbb{M}\,,\\
V \mapsto \xi(P+V)-Q,
\end{eqnarray*}
depending only on $P$ and $Q$, which, in turn, determine $\xi$.
The map $\Upsilon$ is linear and defines a metric compatible affine
connection with vanishing curvature and nonvanishing torsion.
Such connections are known
as Weitzenb\"ock (teleparallel) connections.

We define the covariant derivative of a vector field as
\[
\frac{\partial v^\alpha}{\partial x^\beta}
+\Upsilon^\alpha{}_{\beta\gamma}v^\gamma
\]
and torsion as
\begin{equation}
\label{definition of torsion}
T^\alpha{}_{\beta\gamma}
:=
\Upsilon^\alpha{}_{\beta\gamma}
-
\Upsilon^\alpha{}_{\gamma\beta}\,.
\end{equation}
It is known \cite[formula (7.34)]{nakahara}
that a metric compatible affine
connection is determined by metric and torsion,
so torsion provides a convenient tensorial description of a connection.

Torsion has three irreducible pieces \cite[formulae (4.1)--(4.4)]{mccrea}
\[
T=T^\mathrm{ax}+T^\mathrm{vec}+T^\mathrm{ten},
\]
\begin{equation}
\label{definition of axial torsion}
T^\mathrm{ax}_{\alpha\beta\gamma}
=
\frac13
(
T_{\alpha\beta\gamma}
+
T_{\beta\gamma\alpha}
+
T_{\gamma\alpha\beta}
),
\end{equation}
\begin{equation}
\label{definition of vector torsion}
T^\mathrm{vec}_{\alpha\beta\gamma}
=
\frac13
(
g_{\alpha\beta}T^\mu{}_{\mu\gamma}
-
g_{\alpha\gamma}T^\mu{}_{\mu\beta}
),
\end{equation}
labelled by the adjectives \emph{axial}, \emph{vector} and \emph{tensor} respectively.
We remind the reader that we raise and lower tensor indices using the metric $g$.

\begin{lemma}
\label{lemma about torsion}
For all three groups
$\mathrm{SG}_0^+$, $\mathrm{SG}_0^-$ and $\mathrm{SG}_m$
torsion is constant and vector torsion is zero.
The corresponding formulae for axial torsion read
\[
(*T^\mathrm{ax}_\pm)_\alpha=
\mp
\,
\frac23
\,
(\,
0\,,\,0\,,\,1\,,\,1
\,)\,,
\]
\[
(*T^\mathrm{ax}_m)_\alpha=
-\,
\frac43
\,
(\,
0\,,\,0\,,\,m\,,0
\,)\,.
\]
\end{lemma}

\begin{proof}
Straightforward calculations give the following expressions
for the nonzero connection coefficients.
\begin{itemize}
\item
For $\mathrm{SG}_0^\pm$
\[
\Upsilon^1{}_{32}=\pm 1,
\qquad
\Upsilon^2{}_{31}=\mp 1,
\]
\[
\Upsilon^1{}_{42}=\pm 1,
\qquad
\Upsilon^2{}_{41}=\mp 1.
\]
\item
For $\mathrm{SG}_m$
\[
\Upsilon^1{}_{42}=2m,
\qquad
\Upsilon^2{}_{41}=-2m.
\]
\end{itemize}
It remains only to substitute the above expressions
into formulae
\eqref{definition of torsion}--\eqref{definition of vector torsion}.
\end{proof}

\section{Explicit formulae for our field equations}
\label{Explicit field equations of nonlinear elasticity}

In this appendix we sketch out an algorithm for the derivation
of the explicit form of the differential operator $E(\varphi)$
introduced in Section~\ref{Nonlinear field equations}.
We will do this for the special case
of a Lagrangian of the form \eqref{L mathcal quadratic} from
Example~\ref{example of Quadratic Lagrangian}
and in Minkowski space.
Throughout this appendix we shall use the notation
$\partial_\alpha=\partial/\partial x^\alpha$.

Substituting \eqref{L mathcal quadratic}
into \eqref{relation between two Ls} we get
\begin{equation}
\label{L quadratic}
L(e_2,e_3,e_4)=\alpha(e_2+e_3+e_4)^2+\beta\,e_2\,.
\end{equation}

To begin with, let us rewrite the scalars $e_3$ and $e_4$ in terms of
$\operatorname{tr}(S^k)$, $k=1,2,3,4$:
\begin{subequations}\label{scalar invariants appendix}
\begin{align}
\label{scalar invariant 3 appendix}
e_3&=
\frac16
\left[
(\operatorname{tr}S)^3
-
3(\operatorname{tr}S)\operatorname{tr}(S^2)
+
2\operatorname{tr}(S^3)
\right],
\\
\label{scalar invariant 4 appendix}
e_4&=
\frac1{24}
\left[
(\operatorname{tr}S)^4
-
6(\operatorname{tr}S)^2\operatorname{tr}(S^2)
+
3(\operatorname{tr}(S^2))^2
+
8(\operatorname{tr}S)\operatorname{tr}(S^3)
-
6\operatorname{tr}(S^4)
\right].
\end{align}
\end{subequations}

Substituting
\eqref{scalar invariant 2},
\eqref{scalar invariant 3 appendix}
and
\eqref{scalar invariant 4 appendix}
into
\eqref{L quadratic}
we get a representation of our Lagrangian $L$
as a linear combination of terms
\begin{equation}
\label{term of Lagrangian}
\prod_{j=1}^kS^{\alpha_j}{}_{\beta_j},
\end{equation}
where
$\{\beta_1,\ldots,\beta_k\}$
is some permutation of
$\{\alpha_1,\ldots,\alpha_k\}$.
The number $k$ takes values from two to eight.
In what follows we write down the contribution to $E(\varphi)$
coming from a single term~\eqref{term of Lagrangian}.

The explicit formula for the strain tensor reads
\begin{equation*}
S^\alpha{}_\beta
=
\partial_\beta A^\alpha
+
\partial^\alpha A_\beta
+
(\partial^\alpha A_\gamma)(\partial_\beta A^\gamma).
\end{equation*}
Variation $A^\alpha(x)\mapsto A^\alpha(x)+\Delta A^\alpha(x)$
gives us
\begin{multline*}
\Delta S^\alpha{}_\beta
=
\partial_\beta(\Delta A^\alpha)
+
\partial^\alpha(\Delta A_\beta)
+
(\partial^\alpha(\Delta A_\gamma))(\partial_\beta A^\gamma)
+
(\partial^\alpha A_\gamma)(\partial_\beta(\Delta A^\gamma))
\\
=
\delta^\alpha{}_\gamma\,
\partial_\beta(\Delta A^\gamma)
+
g_{\beta\gamma}\,
\partial^\alpha(\Delta A^\gamma)
+
(\partial^\alpha(\Delta A^\gamma))(\partial_\beta A_\gamma)
+
(\partial^\alpha A_\gamma)(\partial_\beta(\Delta A^\gamma)).
\end{multline*}
We define the linear differential operator
\begin{equation*}
D^\alpha{}_{\beta\gamma}
:=
[g_{\beta\gamma}+(\partial_\beta A_\gamma)]\,\partial^\alpha
+
[\delta^\alpha{}_\gamma+(\partial^\alpha A_\gamma)]\,\partial_\beta
+
2
(\partial^\alpha\partial_\beta A_\gamma).
\end{equation*}
The contribution to $E(\varphi)$
coming from \eqref{term of Lagrangian}
reads
\begin{equation*}
\label{contribution to E}
-
\sum_{l=1}^k
D^{\alpha_l}{}_{\beta_l\gamma}
\prod_{
\substack{j=1\\
j\ne l}}^kS^{\alpha_j}{}_{\beta_j}.
\end{equation*}

The above algorithm can be easily generalised to spacetimes with
$x$-dependent metric and to Lagrangians of general form.

\end{appendices}


\end{document}